\documentclass[11pt,review=true]{acmart}
\usepackage{acm-ec-21}

\usepackage{algpseudocode}
\usepackage{algorithm}
\algtext*{EndWhile}
\algtext*{EndIf}
\algtext*{EndFor}

\usepackage{tikz}
\usetikzlibrary{arrows.meta}

\usepackage{amsmath, amsfonts, amsthm, bbold}
\usepackage{thmtools} 
\usepackage{thm-restate}

\usepackage{hyperref}

\theoremstyle{plain}

\newtheorem{theorem}{Theorem}
\newtheorem{lemma}[theorem]{Lemma}

\theoremstyle{definition}
\newtheorem{definition}[theorem]{Definition}

\usepackage[shortlabels]{enumitem}
\newlist{algoindent}{itemize}{9}
\setlist[algoindent]{label=,leftmargin=.5cm}

\usepackage{mathdots}

\renewcommand{\P}{\mathbb P}

\newcommand{\E}{\mathbb E}
\newcommand{\D}{\mathcal D}
\newcommand{\M}{\mathcal M}
\newcommand{\W}{\mathcal W}

\title{Two-sided matching markets with correlated random preferences}

\author{Hugo Gimbert}
\email{hugo.gimbert@cnrs.fr}
\affiliation{
   \institution{Université de Bordeaux, LaBRI, CNRS, F-33400 Talence}
   \country{France}
}

\author{Claire Mathieu}
\email{clairemmathieu@gmail.com}
\affiliation{
   \institution{Université de Paris, IRIF, CNRS, F-75013 Paris}
   \country{France}
}

\author{Simon Mauras}
\email{simon.mauras@irif.fr}
\affiliation{
   \institution{Université de Paris, IRIF, CNRS, F-75013 Paris}
   \country{France}
}

\begin{abstract}
Stable matching in a community consisting of men and women is a classical combinatorial problem that has been the subject of intense theoretical and empirical study since its introduction in 1962 in a seminal paper by Gale and Shapley, who designed the celebrated ``deferred acceptance'' algorithm for the problem.
\smallbreak
In the input, each participant ranks participants of the opposite type, so the input consists of a collection of permutations, representing the preference lists. A bipartite matching is unstable if some man-woman pair is blocking: both strictly prefer each other to their partner in the matching. Stability is an important economics concept in matching markets from the viewpoint of manipulability. The unicity of a stable matching implies non-manipulability, and near-unicity implies limited manipulability, thus these are mathematical properties  related to the quality of stable matching algorithms. 
\smallbreak
This paper is a theoretical study of the effect of correlations on approximate manipulability of stable matching algorithms. Our approach is to go beyond worst case, assuming that some of the input preference lists are drawn from a distribution. Our model encompasses a discrete probabilistic process inspired by a popularity model introduced by Immorlica and Mahdian, that provides a way to capture correlation between preference lists.  Approximate manipulability is approached from several angles : when all stable partners of a person have approximately the same rank; or when most persons have a unique stable partner. Another quantity of interest is a person's number of stable partners.  Our results aim to study stable matchings in a ``beyond worst case''  setting.
\end{abstract}

\begin{document}

\begin{titlepage}
\maketitle
\end{titlepage}

\section{Introduction}
\label{section:intro}
In the classical stable matching problem, a certain community consists of men and women (all heterosexual and monogamous) where each person ranks those of the opposite sex in accordance with his or her preferences for a marriage partner (possibly declaring some matches as unacceptable).
Our objective is to marry off the members of the community in such a way that the established matching is \emph{stable}, \textit{i.e.} such that there is no \emph{blocking pair}. A man and a woman who are not married to each other form a blocking pair if they prefer each other to their mates. 

In their seminal paper, Gale and Shapley \cite{gale1962college} designed the \emph{men-proposing deferred acceptance} procedure, where men propose while women disposes. This algorithm always outputs a matching which is stable, optimal for men and pessimal for women (in terms of rank of each person's partner). By symmetry, there also exists a women-optimal/men-pessimal stable matching.
Gale and Shapley's original motivation was the assignment of students
to colleges, a setting to which the algorithm and results extend, and their approach was successfully implemented in many
matching markets; see for example
\cite{abdulkadirouglu2005new,abdulkadirouglu2005boston,roth1999redesign,correa2019school}.

However, there exists instances where the men-optimal and women-optimal stable matchings are different, and even extreme cases of instances in which every man/woman pair belongs to some stable matching. This raises the question of 
which matching to choose \cite{gusfield1987three,gusfield1989stable} and of 
possible strategic behavior \cite{dubins1981machiavelli, roth1982economics,demange1987further}.
More precisely, if a woman lies about her preference list, this gives rise to new stable matchings, where she will be no
better off than she would be in the true women-optimal matching.
Thus, a woman can only gain from strategic manipulation up to the maximum difference between her best and worst partners in stable matchings.
By symmetry, this also implies that the men proposing deferred acceptance
procedure is strategy-proof for men (as they will get their best possible
partner by telling the truth).

Fortunately, there is empirical evidence that in many instances, in practice the stable matching is essentially unique (a phenomenon often referred to as {``core-convergence''}); see for example 
\cite{roth1999redesign,pathak2008leveling,hitsch2010matching,banerjee2013marry}.
One of the empirical explanations for core-convergence given by Roth and Peranson in
\cite{roth1999redesign} is that the preference lists are correlated:
\emph{``One factor that strongly influences the size of
the set of stable matchings is the correlation of preferences among programs and
among applicants. When preferences are highly
correlated (i.e., when similar programs tend to
agree which are the most desirable applicants,
and applicants tend to agree which are the most
desirable programs), the set of stable matchings
is small.''}

Following that direction of enquiry, we study the core-convergence phenomenon, in a model where preferences are stochastic. When preferences of women are strongly correlated, Theorem~\ref{thm:geom} shows that the expected difference of rank between each woman's worst and best stable partner is a constant, hence the incentives to manipulate are limited. If additionally the preferences of men are uncorrelated, Theorem~\ref{thm:geomunif} shows that most women have a unique stable partner, and therefore have no incentives to manipulate. Finally, we study the number of stable partners:
when preference lists are drawn from popularity distributions~\cite{immorlica2015incentives,kojima2009incentives,ashlagi2020tiered}, Theorem~\ref{thm:popularity} and~\ref{theorem:womanbounded} give logarithmic upper-bounds on the number of stable partners, matching the lower bound when preferences are uniformly random \cite{knuth1990stable,pittel1992likely}.

\subsection{Definitions and main theorems}

\emph{Matchings.}
Let $\M=\{ m_1,\ldots , m_M\}$ be a set of $M$ men,  $\W=\{ w_1,\ldots , w_W\}$ be a set of $W$ women, and  $N = \min(M, W)$. In a matching, each person is either single, or matched with someone of the opposite sex. Formally, we see a matching as a function $\mu :\M \cup \W \rightarrow \M \cup \W$, which is self-inverse ($\mu^2 = \text{Id}$), where each man $m$ is paired either with a woman or himself ($\mu(m) \in \W\cup\{m\}$), and symmetrically, each woman $w$ is paired with a man or herself ($\mu(w) \in \M \cup\{w\}$). 

\emph{Preference lists.}
Each person declares which members of the opposite sex they find acceptable, then gives a strictly ordered {preference list} of those members. Preference lists are \emph{complete} when no one is declared unacceptable. Formally, we represent the preference list of a man $m$ as a total order $\succ_m$ over $\W\cup\{m\}$, where $w \succ_m m$ means that man $m$ finds woman $w$ acceptable, and $w \succ_{m} w'$ means that man $m$ prefers woman $w$ to woman $w'$. Similarly we define the preference list $\succ_w$ of woman $w$.

\emph{Stability.}
A man-woman pair $(m,w)$ is blocking a matching $\mu$ when $m\succ_w\mu(w)$ and $w\succ_m\mu(m)$. Abusing notations, observe that $\mu$ matches a person $p$ with an unacceptable partner when $p$ would prefer to remain single, that is when the pair $(p,p)$ is blocking. A matching with no blocking pair is stable. A stable pair is a pair which belongs to at least one stable matching

\emph{Random preferences.}
We consider a model where each person's set of acceptable partners is deterministic, and orderings of acceptable partners are drawn independently from \emph{regular} distributions. When unspecified, someone's acceptable partners and/or their ordering is \emph{adversarial}, that is chosen by an adversary who knows the input model but does not know the outcome of the random coin flips.

\begin{definition}[Regular distribution]\label{def:regular}
    A distribution of preferences lists is \textbf{regular} when for every sequence of acceptable partners $a_1, \dots, a_k$ we have $\P[a_1 \succ a_2\,|\,a_2 \succ \dots \succ a_k] \leq \P[a_1 \succ a_2]$. 
\end{definition}

Intuitively, knowing that $a_2$ is ranked well only decreases the probability that $a_1$ beats $a_2$. 
Most probability distributions that have been studied are regular. In particular, sorting acceptable partners by scores (drawn independently from distributions on $\mathbb R$), yields a regular distribution.
As an example of regular distribution, we study popularity preferences, introduced by Immorlica and Madhian \cite{immorlica2015incentives}.

\begin{definition}[Popularity preferences]
\label{definition:popularity}
When a woman $w$ has {\em popularity preferences}, she gives a positive popularity $\D_w(m)$ to each acceptable partner $m$. We see $\D_w$ as a distribution over her acceptable partners, scaled so that it sums to 1. She uses this distribution to draw her favourite partner, then her second favourite, and so on until her least favourite partner.
%
\end{definition}



The following Theorem shows that under some assumptions every woman gives approximately the same rank to all of her stable partners.

\begin{restatable}{mytheorem}{thmgeom}
    \label{thm:geom}
    Assume that each woman  independently draws her  preference list from a regular distribution. The men's preference lists are arbitrary.
    Let $u_k$ be an upper bound on the odds that man $m_{i+k}$ is ranked before man $m_i$:
    \[\forall k\geq 1,\quad u_k = \max_{w,i}\left\{
    \frac{\P[m_{i+k} \succ_w m_i]}{\P[m_{i} \succ_w m_{i+k}]}
    \;\middle|\;\begin{array}{l}
    w \text{ finds both }m_i\text{ and }m_{i+k}\text{ acceptable}
    \end{array}\right\}\]
    Then for each woman with at least one stable partner, in expectation all of her stable partners are ranked within   $(1+2\exp(\sum_{k\geq 1} ku_k)) \sum_{k\geq 1} k^2 u_k$ of one another in her preference list.  
\end{restatable}

Theorem~\ref{thm:geom} is most relevant when the women's preference lists are strongly correlated, that is, when every woman's preference list is ``close"  to a single  ranking $m_1 \succ m_2 \succ \ldots \succ m_M$. This closeness is measured by the odds that in some ranking, some man is ranked ahead of a man who, in the ranking $m_1 \succ m_2 \succ \ldots \succ m_M$, would be $k$ slots ahead of him. 

We detail below three examples of applications, where the expected difference of ranks between each woman's best and worst partners is $O(1)$, and thus her incentives to misreport her preferences are limited.
\begin{itemize}
\item \emph{Identical preferences.} If all women rank their acceptable partners using a master list $m_1 \succ m_2 \succ \dots \succ m_M$, then all $u_k$'s are equal to 0. Then Theorem~\ref{thm:geom} states that each woman has a unique stable husband, a well-known result for this type of instances. 

\item \emph{Preferences from identical popularities.} Assume that women have popularity preferences (Definition~\ref{definition:popularity}) and that each woman gives man $m_i$ popularity $2^{-i}$.
Then $u_k = 2^{-k}$ and the expected rank difference is at most $\mathcal O(1).$

\item \emph{Preferences from correlated utilities.} Assume that women have similar preferences: each woman $w$ gives man $m_i$ a score that is the sum of a common value $i$ and an idiosyncratic value $\eta_i^w$ which is normally distributed with mean $0$ and variance $\sigma^2$; she then sorts men by increasing scores.
Then 
$u_k \leq \max_{w,i}\,\{2\cdot\P[\eta_i^w-\eta_{i+k}^w > k]\} \leq 2e^{-(k/2\sigma)^2}$ and
the expected rank difference, by a short calculation, 
is at most $4\sqrt{\pi}\sigma^3(1+2e^{4\sigma^2}) = \mathcal O(1)$.
\end{itemize}


 A stronger notion of approximate incentive compatibility is near-unicity of a stable matching, meaning that  most persons have either no or one unique stable partner, and thus have no incentive to misreport their preferences. 
 When does that hold? One answer is given by Theorem~\ref{thm:geomunif}.

\begin{restatable}{mytheorem}{thmgeomunif}
    \label{thm:geomunif}
    Assume that each woman  independently draws her preference list from a regular distribution. 
    Let $u_k$ be an upper bound on the odds that man $m_{i+k}$ is ranked before man $m_i$:
    \[\forall k\geq 1,\quad u_k = \max_{w,i}\left\{
    \frac{\P[m_{i+k} \succ_w m_i]}{\P[m_{i} \succ_w m_{i+k}]}
    \;\middle|\;\begin{array}{l}
    w \text{ finds both }m_i\text{ and }m_{i+k}\text{ acceptable}
    \end{array}\right\}\]
    Further assume that all preferences are complete, that $u_k=\exp(-\Omega(k))$, and that men have uniformly random preferences. Then, in expectation the fraction of persons who have multiple stable partners converges to 0.
\end{restatable}

\vspace{-.2cm}
Notice that in the three examples of Theorem~\ref{thm:geom}, the sequence $(u_k)_{k \geq 1}$ is exponentially decreasing.
The assumptions of Theorem~\ref{thm:geomunif} are minimal in the sense that removing one would bring us back to a case where a constant fraction of woman have multiple stable partners.
\begin{itemize}
    \item \emph{Preference lists of women.} If we remove the assumption that $u_k$ is exponentially decreasing, the conclusion no longer holds: consider a balanced market  balanced ($M=W$) in which both men and women have complete uniformly random preferences; then most women have  $\sim \ln N$ stable husbands 
    \cite{knuth1990stable,pittel1992likely}
    \item \emph{Preference lists of men.} Assume that men have random preference built as follows: starting from the ordering $w_1, w_2, \dots, w_M$, each pair $(w_{2i-1},w_{2i})$ is swapped with probability 1/2, for all $i$. A symmetric definition for women's preferences satisfy the hypothesis of Theorem~\ref{thm:geomunif}, with $u_1=1$ and $u_k=0$ for all $k\geq 2$. Then there is a 1/8 probability that men  $m_{2i-1}$ and $m_{2i}$ are both stable partners of women $w_{2i-1}$ and  $w_{2i}$, for all $i$, hence a constant expected fraction of persons with multiple stable partners.
    \item \emph{Incomplete preferences.} Consider a market divided in groups $\{m_{2i-1},m_{2i},w_{2i-1},w_{2i}\}$, where a man and a woman are mutually acceptable if they belong to the same group. Once again, with constant probability, $m_{2i-1}$ and $m_{2i}$ are both stable partners of women $w_{2i-1}$ and  $w_{2i}$.
\end{itemize}

One of the basic results in the study of stable matchings beyond worst case, initiated by Knuth Motwani and Pittel, concerns the number of stable husbands when all preference lists are uniformly random \cite{knuth1990stable}. 


\vspace{-.2cm}
\begin{restatable}{mytheorem}{thmpopularity}
    \label{thm:popularity}
    Let $w$ be a woman.  Assume that $w$ has popularity preferences defined by $\D_w$ and that she has at least one stable partner. The preference lists of  the men and of the women other than $w$ are arbitrary. Then 
  \[\mathbb E[\hbox{Number of stable husbands of }w] \leq 1 + \ln d_{w} + \mathbb E\left[\ln\frac{\D_w(\mu_\W(w))}{\D_w(\mu_\M(w))}\right],\]
where $d_{w}$ denotes the number of acceptable husbands of $w$, $\mu_\M(w)$ is her worst stable partner and $\mu_\W(w)$ is her best stable partner.
\end{restatable}

Theorem~\ref{thm:popularity} generalizes the following bound from \cite{knuth1990stable,pittel1992likely}: assume that $N$ men and $N$ women have complete uniformly random preference lists (``uncorrelated preferences"); then the expected number of stable pairs is asymptotically equivalent to  $N \ln N$. 

\begin{restatable}{corollary}{thmtelescopic}
\label{thm:telescopic}
In both settings (1) and (2), the expected number of stable pairs is at most $N(1 + \ln N)$
\begin{enumerate}
    \item[(1)] Assume that women have popularity preferences and that each man $m$ has an \emph{intrinsic popularity} $\D(m)$, such that if woman $w$ finds man $m$ acceptable then $\D_w(m) = \D(m)$.
    \item[(2)] Assume that men and women have symmetric \emph{popularity preferences}, such that if man $m$ and woman $w$ are mutually acceptable then $\D_m(w) = \D_w(m)$.
\end{enumerate}
\end{restatable}

\vspace{-.2cm}
Intrinsic popularities model ``one-sided'' correlations, for example when all women agree that some men are more popular.
Symmetric popularities model ``cross-sided'' correlations, for example when men and women prefer partners with whom they share some centers of interest. 
Both intrinsic popularities and symmetric popularities generalizes the uniform case.
The upper bound from Corollary~\ref{thm:telescopic} matches the bound from~\cite{pittel1992likely}, implying that uncorrelated preferences are a worst case situation up to lower order terms: correlations reduce the number of stable pairs.

\vspace{-.2cm}\begin{proof} For the purpose of this introduction, we assume that $M = W$ and that everyone has complete preferences, which implies that everyone is matched in all stable matchings. The proof without this assumption is more technical and can be found in Appendix~\ref{section:telescopic}.
In both settings~(1) and~(2), Theorem~\ref{thm:popularity} can be applied for each woman, therefore
\[\textstyle\E[\text{nb of stable pairs}] \leq N(1+\ln N) + \E[\sum_{w\in\W}\ln\D_w(\mu_\W(w))-  \sum_{w\in\W}\ln\D_w(\mu_\M(w))].\]
In setting (1) both sums are equal to $\sum_{m\in\M}\ln\D(m)$, thus their difference is equal to 0, concluding the proof. In setting~(2), we swap the roles of men and women, and apply Theorem~\ref{thm:popularity} for each man,
\[\textstyle\E[\text{nb of stable pairs}] \leq N(1+\ln N) + \E[\sum_{m\in\M}\ln\D_m(\mu_\M(m))-  \sum_{m\in\M}\ln\D_m(\mu_\W(m))].\]
Now, observe that for every matching $\mu$ we have $\sum_{w\in\W}\ln\D_w(\mu(w)) = \sum_{m\in\M}\ln\D_m(\mu(m))$, thus summing the two formula concludes the proof. 
\end{proof}

Recall that the upper bound from Theorem~\ref{thm:popularity} depends on a ratio of popularities.
When $w$ has complete popularity preferences and each man $m_i$ has popularity $0.99^i$, the ratio is at most $0.99^{-M}$, thus at most $\approx 1\%$ of men are stable husbands of $w$. In Appendix~\ref{section:onewomanlb}, we show that this $1\%$ upper bound is tight in the worst case.

%
%
%
%
%
%
When everyone has popularity preferences, the following theorem provides an improved upper bound on the popularity ratio between the stable husbands of a woman. 
The bound depends on how uniform the preferences of men are (parameter $R_\M$, the maximal ratio between the popularities of two distinct women for a given man)
and how similar the preferences of women are (parameter $Q_\W$, the maximal ratio between the popularities of two distinct men for two distinct women).
The parameter  
$R_\M$ is small when the preferences of every man among women are close to be uniform; they are actually uniform when $R_\M=1$.
The parameter  
$Q_\W$ is close to $1$ when women tend to agree on the relative popularities of men. In case  men have intrinsic popularities, like in setting (1) of Corollary~\ref{thm:telescopic}, then $Q_\W=1$.



\vspace{-.2cm}
\begin{restatable}{mytheorem}{womanboundedtheorem}
\label{theorem:womanbounded}
Assume that men and women have popularity preferences.
Denote
\[
R_\M = \max_{ \substack{m\in \M\\w_0,w_1 \in \W}}
\frac{\D_m(w_0)}{\D_m(w_1)}
\qquad
Q_\W = \max_{\substack{w_0,w_1 \in \W\\m_0,m_1\in \M}}
\frac{\D_{w_0}(m_0)}{\D_{w_0}(m_1)}
\cdot
\frac{\D_{w_1}(m_1)}{\D_{w_1}(m_0)}\enspace.
\]
Let $w$ be a woman.
Then with probability 
$\geq (1 - \frac{2}{N^2})$ the popularity ratio (for $w$) between any two stable husbands of $w$ is no more than
$
\left(N^5\cdot Q_\W\right)^ {
1 + 
\frac{4\ln(N)\left(1 + \log_2(N)\right)}{\ln\left(1 + 1/R_\M\right)}
}\enspace.
$
\end{restatable}

From this upper-bound bound on the popularities of stable husbands,
one derives an upper-bound on the expected number of stable husbands.

\vspace{-.2cm}
\begin{restatable}{corollary}{womanboundedcor}
\label{cor:womanbounded}
Assume that men and women have popularity preferences.
The expected number of stable husbands of any woman $w$
is  bounded by $\mathcal{O}\bigg(\frac{\ln(Q_\W)}{\ln\left(1 + \frac{1}{R_\M} \right)}\ln^3(N)\bigg)\enspace.$
\end{restatable}

Corollary~\ref{cor:womanbounded}  is most relevant when $R_\M$ and $Q_\W$ are not too large: the bound is polylog in $N$ as long as $R_\M$ is polylog in $N$ and $Q_\W$ is polynomial in $N$.

In the extreme case where both $R_\M=1$ and $Q_\M=1$,
by symmetry all women have the same expected number 
of stable husbands;
and Corollary~\ref{thm:telescopic} shows that this number is $\ln(N)$ while Corollary~\ref{cor:womanbounded}
provides a looser upper-bound of $\ln^3(N)$.



\vspace{-.2cm}\subsection{Related work}

Analyzing instances that are less far-fetched than in the worst case is the motivation underlying the model of stochastically generated preference lists.
A series of papers \cite{pittel1989average,knuth1990stable,pittel1992likely,pittel2007number,lennon2009likely} study the model where $N$ men and $N$ women have complete uniformly random preferences. Asymptotically, and in expectation, the total number of stable matchings is $\sim e^{-1}N \ln N$, in which a fixed woman has $\sim \ln N$ stable husbands, where her best stable husband has rank $\sim\ln N$ and her worst stable husband has rank $\sim N/\ln N$. Theorem~\ref{thm:popularity} and its proof extend the upper-bound on the number of stable husbands from \cite{knuth1990stable}.

The first theoretical explanations of the ``core-convergence'' phenomenon where given in \cite{immorlica2015incentives} and \cite{ashlagi2017unbalanced}, in variations of the standard uniform model. Immorlica and Mahdian \cite{immorlica2015incentives} consider the case where men have constant size random preferences (truncated popularity preferences). Ashlagi, Kanoria and Leshno \cite{ashlagi2017unbalanced}, consider slightly unbalanced matching markets ($M < W$). Both articles prove that the fraction of persons with several stable partners tends to 0 as the market grows large. Theorem~\ref{thm:geomunif} and its proof incorporate ideas from those two papers.

Beyond strong ``core-convergence'', where most agents have a unique stable partner, one can bound the utility gain by manipulating a stable mechanism. Lee \cite{lee2016incentive} considers a model with random cardinal utilities, and shows that agents receive almost the same utility in all stable matchings.
Kanoria, Min and Qian \cite{kanoria2021matching}, and Ashlagi, Braverman, Thomas and Zhao \cite{ashlagi2020tiered} study the rank of each person's partner, under the men and women optimal stable matchings, as a function of the market imbalance and the size of preference lists \cite{kanoria2021matching}, or as a function of each person's (bounded) popularity \cite{ashlagi2020tiered}. Theorem~\ref{thm:geom} can be compared with such results.

Beyond one-to-one matchings, school choice is an example of many-to-one markets. Kojima and Pathak \cite{kojima2009incentives} generalize results from~\cite{immorlica2015incentives} and prove that most schools have no incentives to manipulate. Azevedo and Leshno \cite{azevedo2016supply} show that large markets converge to a unique stable matching in a model with a continuum of students. To counter balance those findings, Bir\'o, Hassidim, Romm and Shorer \cite{biro2020need}, and Rheingans-Yoo \cite{rheingans2020large} argue that socioeconomic status and geographic preferences might undermine core-convergence, thus some incentives remain in such markets.
\vspace{-.2cm}


\section{Strongly correlated preferences: Proof of Theorem~\ref{thm:geom}}
\label{section:geom}
\newcommand{\MOSM}{\mathcal H}

\thmgeom*



\vspace{-.1cm} In subsection \ref{subsection:separators}, we define a partition of stable matching instances into \emph{blocks}. For strongly correlated instances, blocks provide the structural insight to start the analysis: in Lemma~\ref{lemma:rankdifference}, we use them to upper-bound the difference of ranks between a woman's worst and best stable partners by the sum of  (1) the number $x$ of men coming from other blocks and who are placed between stable husbands in the woman's preference list, and (2) the block size.  

The analysis requires a delicate handling of conditional probabilities. In subsection \ref{subsection:MOSM}, we explain how to condition on the men-optimal stable matching, when preferences are random.

Subsection~\ref{subsection:tractable} analyzes (1). The men involved are out of place compared to their position in the ranking $m_1\succ \ldots \succ m_M$, and the odds of such events can be bounded, thanks to the assumption  that distributions of preferences are {regular}. Our main technical lemma there is Lemma~\ref{lemma:regular}. 

Subsection~\ref{subsection:block} analyzes (2), the block size by first giving a simple greedy algorithm (Algorithm~\ref{algo:separators}) to compute a block. Each of the two limits of a block is computed by a sequence of ``jumps", so the total distance traveled is a sum of jumps which, thanks to Lemma~\ref{lemma:regular} again, can be stochastically dominated by a sum $X$ of independent random variables (see 
Lemma~\ref{lemma:domination}
); thus it all reduces to analyzing $X$, a simple mathematical exercise (Lemma~\ref{lemma:meandomination}).

Finally, subsection~\ref{subsection:puttingeverythingtogether} combines the Lemmas previously established to prove Theorem~\ref{thm:geom}. 



Our analysis builds on Theorems~\ref{lemma:MPDA} and ~\ref{thm:rural}, two fundamental and well-known 
results. 

\begin{theorem}[Adapted from \cite{gale1962college}]
\label{lemma:MPDA}
Algorithm \ref{algo:MPDA} outputs a stable matching $\mu_\M$ in which every man (resp. woman) has his best (resp. her worst) stable partner. Symmetrically, there exists a stable matching $\mu_\W$ in which every woman (resp. man) has her best (resp. his worst) stable partner.
\end{theorem}

\begin{theorem}[Adapted from \cite{gale1985some}]
\label{thm:rural}
 Each person is either matched in all stable matchings, or single in all stable matchings. In particular, a woman is matched in all stable matchings if and only if she received at least one acceptable proposal during Algorithm~\ref{algo:MPDA}.
\end{theorem}

\vspace{-.2cm}
\begin{algorithm}[h!]
\begin{algoindent}
\item\textbf{Input:}
Preferences of men $(\succ_m)_{m \in M}$ and women $(\succ_w)_{w \in \W}$.
\item\textbf{Initialization :} Start with an empty matching $\mu$. 
\item\textbf{While} a man $m$ is single and has not proposed to every woman he finds acceptable\textbf{, do}
\begin{algoindent}
  \item $m$ proposes to his favorite woman $w$ he has not proposed to yet.
  \item \textbf{If} $m$ is $w$'s favorite acceptable man among all proposals she received\textbf{,}
  \begin{algoindent}
  \item $w$ accepts $m$'s proposal, and rejects her previous husband if she was married.\\
  \end{algoindent}
\end{algoindent}
\item\textbf{Output:} Resulting matching.
\end{algoindent}
\caption{Men Proposing Deferred Acceptance.}
\label{algo:MPDA}
\end{algorithm}
\vspace{-.5cm}

\subsection{Separators and blocks}
\label{subsection:separators}
In this subsection, we define the block structure underlying our analysis.


\vspace{-.2cm}
\begin{definition}[separator]
A \emph{separator} is a set $S \subseteq \M$ of men such that in the men-optimal stable matching  $\mu_\M$, each woman married to a man in $S$ prefers him to all men outside $S$ : 
    \[\forall w \in \mu_\M(S)\cap\W,\quad \forall m \in\M\setminus S,\quad \mu_\M(w) \succ_w m\] 
\end{definition}

\vspace{-.2cm}
\begin{lemma}\label{lemma:separator}
    Given a separator $S \subseteq \M$, every stable matching matches $S$ to the same set of women.
\end{lemma}

\vspace{-.2cm}
\begin{proof}
Let $w \in \mu_\M(S)$ and let $m$ be the partner of $w$ in some stable matching. Since $\mu_\M$ is the woman-pessimal stable matching by Theorem~\ref{lemma:MPDA}, $w$ prefers $m$ to $\mu_\M(w)$. By definition of separators, that implies that $m\in S$. Hence, in every stable matching $\mu$, women of $\mu_\M(S)$ are matched to men in $S$. By a cardinality argument, men of $S$ are matched by $\mu$ to $\mu_\M(S)$.
\end{proof}

\vspace{-.2cm}
\begin{definition}[prefix separator, block]
    \label{def:block}
    A \emph{prefix separator} is a separator $S$ such that  $S = \{m_1, m_2, \dots, m_t\}$ for some $0 \leq t \leq N$.     
    Given a collection of $b+1$ prefix separators $S_i = \{m_1, \dots, m_{t_i}\}$ with $0 = t_0 < t_1 < \dots < t_b = N$, the $i$-th \emph{block} is the set $B_i = S_{t_i} \setminus S_{t_{i-1}}$ with $1 \leq i \leq b$.
    
    Abusing notations, we will denote $S$ as the prefix separator $t$ and $B$ as the block $(t_{i-1},t_i]$. 
\end{definition}

\begin{lemma}\label{lemma:block}
    Given a block $B \subseteq \M$, every stable matching matches $B$ to the same set of women.
\end{lemma}
\begin{proof}
    $B$ equals $S_{t_i} \setminus S_{t_{i-1}}$ for some $i$. Applying Lemma~\ref{lemma:separator} to $S_{t_i}$ and to $S_{t_{i-1}}$ proves the Lemma.
\end{proof}

\vspace{-.3cm}
\begin{lemma}\label{lemma:rankdifference}
    Consider a woman $w_n$  who is matched by $\mu_\M$ and let $B=(l,r]$ denote her block. Let $x$ denote the number of men from a better block that are ranked by $w_n$ between a man of $B$ and $m_n$:  
\vspace{-.2cm}
    $$x = |\{i \leq l\;|\; \exists j > l,\;m_j \succ_{w_n} m_i \succ_{w_n} m_n\}|.
\vspace{-.2cm}$$
    Then in $w_n$'s preference list, the difference of ranks between $w_n$'s worst and best stable partners is at most $x+r-l-1$.
\end{lemma}

\vspace{-.3cm}
\begin{proof}
Since $\mu_\M$ is woman-pessimal by Theorem~\ref{lemma:MPDA}, $m_n$ is the last stable husband in $w_n$'s preference list. Let $m_j$ denote her best stable husband. 

In $w_n$'s preference list, the interval from $m_j$ to $m_n$ contains men from her own block, plus possibly some additional men. Such a man $m_i$ comes from outside her block $(l,r]$  and she prefers him to $m_n$: since $r$ is a prefix separator, we must have $i\leq l$. Thus $x$ counts the number of men who do not belong to her block but who in her preference list are ranked between $m_j$ and $m_n$. 

On the other hand, the number of men who belong to her block and who in her preference list are ranked between $m_j$ and $m_n$ (inclusive) is at most $r-l$.

Together, the difference of ranks between $w_n$'s worst and best stable partners is at most $x+(r-l)-1$. See Figure~\ref{fig:pref} for an illustration. 
\end{proof}
\vspace{-.5cm}

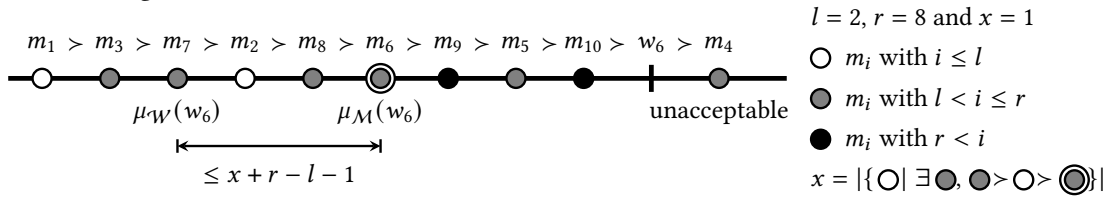
\begin{figure}[h!]
\vspace{-.2cm}
\begin{center}
\scalebox{.9}{\begin{tikzpicture}[line width=2pt,>=stealth]
    \def\xscale{1}
    \draw (.5*\xscale,0) -- (12*\xscale,0);
    \foreach \i/\t in {1/$m_1$,2/$m_3$,3/$m_7$,4/$m_2$,5/$m_8$,
    6/$m_6$,7/$m_9$,8/$m_5$,9/$m_{10}$,10/$w_6$,11/$m_4$} {
        \draw (\i*\xscale,0.5) node[anchor=center] {\t};
    }
    \fill[white] (6*\xscale,0) circle (6pt);
    \draw[line width=1pt] (6*\xscale,0) circle (6pt);
    \foreach \i in {1.5,2.5,...,10.5} {
        \draw (\i*\xscale,.5) node[anchor=center] {$\succ$};
    }
    \foreach \i in {1,4} {
        \fill[white](\i*\xscale,0) circle (4pt);
        \draw[line width=1pt] (\i*\xscale,0) circle (4pt);
    }
    \foreach \i in {2,3,5,6,8,11} {
        \fill[gray](\i*\xscale,0) circle (4pt);
        \draw[line width=1pt] (\i*\xscale,0) circle (4pt);
    }
    \foreach \i in {7,9} {
        \fill[black](\i*\xscale,0) circle (4pt);
        \draw[line width=1pt] (\i*\xscale,0) circle (4pt);
    }
    \draw (10*\xscale,-.2) -- (10*\xscale,.2);
    \draw (3*\xscale,-.5) node[anchor=center] {$\mu_\W(w_6)$};
    \draw (6*\xscale,-.5) node[anchor=center] {$\mu_\M(w_6)$};
    \draw (11*\xscale,-.5) node[anchor=center] {unacceptable};
    \draw[<->, line width=1pt] (3*\xscale,-1) -- (6*\xscale,-1);
    \draw[line width=1pt] (3*\xscale,-1.1) -- (3*\xscale,-0.9);
    \draw[line width=1pt] (6*\xscale,-1.1) -- (6*\xscale,-0.9);
    \draw (4.5*\xscale,-1.4) node[anchor=center] {$\leq x +r-l-1$};
    \fill[white](12.5*\xscale,0.3) circle (4pt);
    \fill[gray](12.5*\xscale,-.3) circle (4pt);
    \fill[black](12.5*\xscale,-.9) circle (4pt);
    \draw[line width=1pt] (12.5*\xscale,0.3) circle (4pt);
    \draw[line width=1pt] (12.5*\xscale,-.3) circle (4pt);
    \draw[line width=1pt] (12.5*\xscale,-.9) circle (4pt);
    \draw (12.2*\xscale,0.9) node[anchor=west] {$l=2$, $r=8$ and $x = 1$};
    \draw (12.7*\xscale,0.3) node[anchor=west] {$m_i$ with $i \leq l$};
    \draw (12.7*\xscale,-.3) node[anchor=west] {$m_i$ with $l < i \leq r$};
    \draw (12.7*\xscale,-.9) node[anchor=west] {$m_i$ with $r < i$};
    \draw (12.2*\xscale,-1.5) node[anchor=west] {$x = |\{\quad|\;\exists\quad,\quad\succ\quad\succ\;\quad\}|$};
    \fill[gray] (14.35*\xscale,-1.5) circle (4pt);
    \fill[gray] (14.85*\xscale,-1.5) circle (4pt);
    \fill[gray] (16.25*\xscale,-1.5) circle (4pt);
    \draw[line width=1pt] (16.25*\xscale,-1.5) circle (4pt);
    \draw[line width=1pt] (16.25*\xscale,-1.5) circle (6pt);
    \draw[line width=1pt] (13.5*\xscale,-1.5) circle (4pt);
    \draw[line width=1pt] (14.35*\xscale,-1.5) circle (4pt);
    \draw[line width=1pt] (14.85*\xscale,-1.5) circle (4pt);
    \draw[line width=1pt] (15.5*\xscale,-1.5) circle (4pt);
\end{tikzpicture}}
\end{center}
\vspace{-.2cm}
\caption{Preference list of $w_n$, with $n=6$. 
The block of $w_n$ is defined by a left separator at $l=2$ and a right separator at $r=8$.
Colors white, gray and black corresponds to blocks, and are defined in the legend.
All stable partners of $w_n$ must be gray.
Men in black are all ranked after $m_n = \mu_\M(w_n)$.
The difference in rank between $w_n$'s worst and best partner is at most the number of gray men (here $r-l = 6$), minus 1, plus the number of white men ranked after a gray man and before $m_n$ (here $x=1$).
}
\vspace{-.7cm}
\label{fig:pref}
\end{figure}

\subsection{Conditioning on the man optimal stable matching when preferences are random}
\label{subsection:MOSM}

We study the case where each person draws her preference list from an arbitrary  distribution. The preference lists are random variables, that are independent but not necessarily identically distributed.

Intuitively, we use the \emph{principle of deferred decision} and construct preference lists in an online manner. By Theorem~\ref{lemma:MPDA} the man-optimal stable matching $\mu_\M$ is computed by Algorithm~\ref{algo:MPDA}, and the remaining randomness can be used for a stochastic analysis of each person's stable partners. To be more formal, we define a random variable $\MOSM$, and inspection of Algorithm~\ref{algo:MPDA} shows that $\MOSM$ contains enough information on each person's preferences to run Algorithm~\ref{algo:MPDA} deterministically.

\vspace{-.2cm}\begin{definition}\label{def:MOSM}
Let $\MOSM = (\mu_\M, (\sigma_m)_{m\in \M}, (\pi_w)_{w\in\W})$ denote the random variable consisting of (1) the man-optimal stable matching $\mu_\M$, (2) each man's ranking of the women he prefers to his partner in $\mu_\M$, and (3) each woman's ranking of the men who prefer her to their partner in $\mu_\M$.
\end{definition}
\vspace{-.2cm}

\subsection{Analyzing the number $x$ of men from other blocks}
\label{subsection:tractable}



\begin{lemma}
    \label{lemma:regular}
    Recall  the sequence $(u_k)_{k\geq 1}$ defined in the statement of Theorem~\ref{thm:geom}:
    \[\forall k\geq 1,\quad u_k = \max_{w,i}\left\{
    \frac{\P[m_{i+k} \succ_w m_i]}{\P[m_{i} \succ_w m_{i+k}]}
    \;\middle|\;\begin{array}{l}
    w \text{ finds both }m_i\text{ and }m_{i+k}\text{ acceptable}
    \end{array}\right\}\]
    Let $w$ be a woman. Given a subset of her acceptable men and  a ranking of that subset  $a_1 \succ_{w} \dots \succ_{w} a_p$, we condition on the event that in $w$'s preference list, $a_1 \succ_{w} \dots \succ_{w} a_p$ holds.
    Let $m_i = a_1$ be $w$'s favorite man in that subset. Let $J_i$ be a random variable, equal to the highest $j \geq i$ such that woman $w$ prefers $m_j$ to $m_i$. Formally, $J_i = \max\{j \geq i\;|\; m_j \succeq_{w} m_i\}$.
    Then, for all $k\geq 1$, we have
    \[
    \P[J_i< i+k\;|\;J_i < i+k+1]
    \geq \exp(-u_{k}),
    \quad\text{and }\quad
    \P[J_i<i+k] \geq \exp(\textstyle-\sum_{\ell \geq k}u_\ell).\]
\end{lemma}
\vspace{-.2cm}
\begin{proof}
    $J_i$ is determined by $w$'s preference list. We construct $w$'s preference list using the following algorithm: initially we know her ranking $\sigma_A$ of the subset $A=\{ a_1,a_2,\ldots , a_p\}$ of acceptable men, and $m_i=a_1$ is her favorite among those. For each $j$ from $N$ to $i$ in decreasing order, we insert $m_j$ into the ranking according to the distribution of $w$'s preference list, stopping as soon as some $m_j$ is ranked before $m_i$ (or when $j=i$ is that does not happen). Then the step $j\geq i$ at which this algorithm stops  equals $J_i$.

    To analyze the algorithm, observe that at each step $j=N, N-1, \ldots  $, we already know $w$'s ranking of the subset  $S=\{m_{j+1}, \dots, m_N\} \cup \{a_1, \dots, a_p\} \cup \{\text{men who are not acceptable to }w\}$. If $m_j$ is already in $S$, $w$ prefers $m_i$ to $m_j$, thus the algorithm continues and $J_i < j$. Otherwise the algorithm inserts $m_j$ into the existing ranking:  by definition of regular distributions (Definition~\ref{def:regular}), the probability that $m_j$ beats $m_i$ given the ranking constructed so far is at most the unconditional probability $\P[m_j \succ_{w} m_i]$. 
    \[
    \P[J_i < j \;|\; w\text{'s partial ranking at step $j$}] \geq
    1 - \P[m_j \succ_{w} m_i].
    \]
   By definition of $u_{j-i}$, we have
    $
    1 - \P[m_j \succ_{w} m_i]=
    \left(\textstyle1+\frac{\P[m_j\succ_w m_i]}{\P[m_i\succ_w m_j]}
    \right)^{-1} \geq (1+u_{j-i})^{-1} \geq \exp(-u_{j-i}).
    $
    ~\smallskip\par\noindent
    Summing over all rankings $\sigma_S$ of $S$ that are compatible with $\sigma_A$ and with $J_i\leq j$,
    $$
    \P[J_i<j\;|\;J_i \leq j]= \hspace{-.5cm}\sum_{
    \substack{\text{$\sigma_S$ compatible with}\\ J_i\leq j \text{ and with }\sigma_A}}\hspace{-.5cm}
    \P[\sigma_S\;|\;\sigma_A]\cdot\P[J_i<j\;|\;\sigma_S]
    \geq \sum_{\sigma_S} \P[\sigma_S\;|\;\sigma_A]\cdot \exp(-u_{j-i}) = \exp(-u_{j-i}).
    $$
Finally,
$
\P[J_i<j]=
\prod_{\ell =j}^N 
\P[J_i<\ell \;|\;J_i \leq \ell]\geq \prod_{k\geq j-i} \exp(-u_{k})= \exp(-\sum_{k\geq j-i} u_{k} ).
$
\end{proof}


Recall from Lemma~\ref{lemma:rankdifference} that $r-l-1+x$ is an upper bound on the difference of rank of woman $w_n$'s worst and best stable husbands. 
We first bound the expected value of the random variable $x$ defined in Lemma~\ref{lemma:rankdifference}.

\begin{lemma}\label{lemma:delta}
    Given a woman $w_n$, define the random variable  $x$ as in Lemma~\ref{lemma:rankdifference}: conditioning on $\MOSM$,  $x = |\{i \leq l\;|\; \exists j > l,\;m_j \succ_{w_n} m_i \succ_{w_n} m_n\}|$ is the number of men in a better block, who can be ranked between $w_n$'s worst and best stable husbands. Then $\E[x] \leq \sum_{k \geq 1} ku_k$.
\end{lemma}

\begin{proof}
    Start by conditioning on $\MOSM$, and let $m_n = a_1 \succ_{w} a_2 \succ_{w} \dots \succ_{w} a_p$ be $w_n$'s ranking of men who prefer her to their partner in $\mu_\M$.
    We draw the preference lists of each woman $w_i$ with $i < n$, and use Algorithm~\ref{algo:separators} to compute the value of $l$. 
    
    For each $i \leq l$, we proceed as follows. If $m_n \succ_{w_n} m_i$, then $m_i$ cannot be ranked between $w_n$'s worst and best stable partners. Otherwise, we are in a situation where $m_i \succ_{w_n} a_1 \succ{w_n} \dots \succ_{w_n} a_p$. Using notations from Lemma~\ref{lemma:regular}, $w$ prefers $m_i$ to all $m_j$ with $j > l$ if and only if $J_i < l+1$. By  Lemma~\ref{lemma:regular} this occurs with probability at least $\exp(-\sum_{k \geq l+1-i} u_k)$. Thus
    \[\textstyle
        \P[\exists j>l,\;m_j \succ_{w_n} m_i \succ_{w_n} m_n\;|\;\MOSM,l] \leq 1 - \exp(-\sum_{k \geq l+1-i} u_k) \leq \sum_{k \geq l+1-i} u_k.
    \]
    Summing this probability for all $i \leq l$, we obtain
    $\E[x\;|\;\MOSM, l] \leq \sum_{i \leq l} \sum_{k \geq l+1-i} u_k \leq \sum_{k \geq 1} k u_k$.
\end{proof}

\subsection{Analyzing the block size}
\label{subsection:block}


\begin{lemma}\label{lemma:correctnessalgorithm2}
    Consider $w_n$  who is matched by $\mu_M$. Then  Algorithm~\ref{algo:separators} outputs the block containing $w_n$.
\end{lemma}

\vspace{-.3cm}
\begin{algorithm}[h!]
    \begin{algoindent}
        \item \textbf{Initialization:}
        \begin{algoindent}
            \item Compute the man optimal stable matching $\mu_\M$.
            \item Relabel women so that $w_i$ denotes the wife of $m_i$ in $\mu_\M$
            \item Pick a woman $w_n$ who is married in $\mu_\M$.
        \end{algoindent}
        \item \textbf{Left prefix separator:} initialize $l \gets n-1$
        \begin{algoindent}
            \item \textbf{while} there exists $i\leq l$ and $j>l$ such that $m_j \succ_{w_i} m_i$:
            \begin{algoindent}
                \item $l \gets \min\{i \leq l\;|\;\exists j>l,\;m_j \succ_{w_i} m_i\} - 1$.
            \end{algoindent}
        \end{algoindent}
        \item \textbf{Right prefix separator:} initialize $r \gets n$.
        \begin{algoindent}
            \item \textbf{while} there exists $j> r$ and $i\leq r$ such that $m_j \succ_{w_i} m_i$:
            \begin{algoindent}
                \item $r \gets \max\{j > r\;|\;\exists i \leq r, \; m_j \succ_{w_i} m_i\}$.
            \end{algoindent}
        \end{algoindent}
        \smallskip\item \textbf{Output:} $(l,r]$.
    \end{algoindent}
    \caption{Computing a block}
    \label{algo:separators}
\end{algorithm}
\vspace{-.3cm}

\begin{proof}
Algorithm~\ref{algo:separators} is understood most easily by following its execution on Figure~\ref{figure:jump}. 
Algorithm~\ref{algo:separators} applies a right-to-left greedy method to find the largest prefix separator $l$ which is $\leq n-1$. By definition of prefix separators, a witness that some $t$ is not a prefix separator is a pair $(m_j,w_i)$ where $j>t\geq i$ and  woman $w_i$ prefers man $m_j$ to her partner: $m_j >_{w_i} m_i$. Then the same pair also certifies that no $t'=t,t-1,t-2,\ldots ,i$ can be a prefix separator either, so the algorithm jumps to $i-1$ and looks for a witness again. When there is no witness, a prefix separator has been found, thus $l$ is the largest prefix separator $\leq n-1$. 
Similarly, Algorithm~\ref{algo:separators} computes the smallest prefix separator $r$ which is $\geq n$. 
Thus, by definition of blocks, $(l,r]$ is the block containing $w_n$. 
\end{proof}

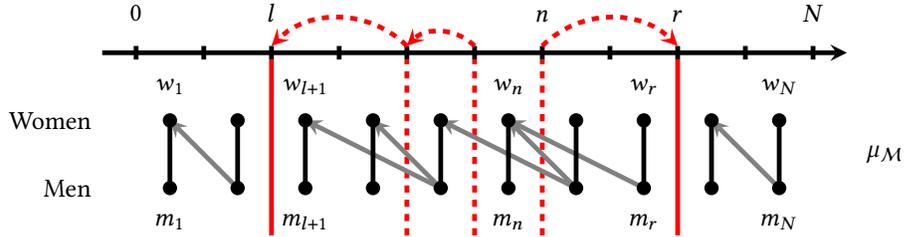
\begin{figure}[h!]
\vspace{-.2cm}
  \begin{center}
  \scalebox{.9}{\begin{tikzpicture}[line width=2pt,>=stealth]
    \def\xscale{1}
    \draw[->] (0*\xscale,3) -- (11*\xscale,3);
    \draw (0,2) node[anchor=east] {Women};
    \draw (12*\xscale,1.5) node[anchor=east] {$\mu_\M$};
    \draw (0,1) node[anchor=east] {Men};
    \foreach \x in {0.5,1.5,...,10.5}
      \draw (\x*\xscale,2.9) -- (\x*\xscale,3.1);
    \foreach \x in {4.5,5.5,6.5}
      \draw[dashed, red] (\x*\xscale, 0.3) -- (\x*\xscale, 2.9);
    \foreach \x in {2.5,8.5}
      \draw[red] (\x*\xscale, 0.3) -- (\x*\xscale, 2.9);
    \draw[->, dashed, red] (6.5*\xscale,3.1) to [bend left=45] (8.5*\xscale,3.1);
    \draw[->, dashed, red] (5.5*\xscale,3.1) to [bend right=45] (4.5*\xscale,3.1);
    \draw[->, dashed, red] (4.5*\xscale,3.1) to [bend right=45] (2.5*\xscale,3.1);
    \foreach \a/\b in {1/2,3/5,4/5,5/7,6/7,6/8,9/10}
      \draw[gray,->] (\b*\xscale, 1) -- (\a*\xscale, 2);
    \draw (1*\xscale, .5) node {$m_{1}$};
    \draw (3*\xscale, .5) node {$m_{l+1}$};
    \draw (6*\xscale, .5) node {$m_{n}$};
    \draw (8*\xscale, .5) node {$m_{r}$};
    \draw (10*\xscale, .5) node {$m_{N}$};
    \draw (1*\xscale, 2.5) node {$w_{1}$};
    \draw (3*\xscale, 2.5) node {$w_{l+1}$};
    \draw (6*\xscale, 2.5) node {$w_{n}$};
    \draw (8*\xscale, 2.5) node {$w_{r}$};
    \draw (10*\xscale, 2.5) node {$w_{N}$};
    \foreach \x in {1,...,10} {
      \fill (\x*\xscale,2) circle (3pt);
      \fill (\x*\xscale,1) circle (3pt);
      \draw[line width=2pt] (\x*\xscale,1) -- (\x*\xscale,2);
    }
    \draw[anchor=south] (\xscale*10.5,3.3) node {$N$};
    \draw[anchor=south] (\xscale*8.5,3.3) node {$r$};
    \draw[anchor=south] (\xscale*6.5,3.3) node {$n$};
    \draw[anchor=south] (\xscale*2.5,3.3) node {$l$};
    \draw[anchor=south] (\xscale*0.5,3.3) node {$0$};
  \end{tikzpicture}}
  \end{center}
  \vspace{-.3cm}
  \caption{Computing the block containing woman $w_n$. The vertical black edges correspond to the men-optimal stable matching $\mu_\M$. There is a light gray arc $(m_j,w_i)$ if $j>i$ and  woman $w_i$ prefers man $m_j$ to her partner: $m_j \succ_{w_i} m_i$. The prefix separators correspond to the solid red vertical lines which do not intersect any gray arc.  Algorithm~\ref{algo:separators} applies a right-to-left greedy method to find the largest prefix separator $l$ which is $\leq n-1$, jumping from dashed red line to dashed red line, and a similar left-to-right greedy method again to find the smallest prefix separator $r$ which is $\geq n$. This determines the block $(l,r]$ containing $n$.}
  \vspace{-.2cm}
  \label{figure:jump}
\end{figure}

\vspace{-.2cm}
\begin{definition}\label{def:domination}
Let $X$ be the random variable defined as follows.  Let $(\Delta_t)_{t \geq 0}$ denote a sequence of i.i.d.r.v.'s taking non-negative integer values with the following distribution:
\[\forall \delta > 0, \quad\P[\Delta_t < \delta] = \exp\left(\textstyle-\sum_{k \geq \delta} k u_k\right)\]
Then $X = \Delta_0+\Delta_1+\cdots +\Delta_{T-1} $, where $T$ is the first $t\geq 0$ such that $\Delta_t=0$.
 \end{definition}
 

The proofs of the following Lemmas can be found in Appendix~\ref{section:geom_proofs}.

\begin{restatable}{lemma}{lemmadomination}\label{lemma:domination}
    Given a woman $w_n$, let $(l,r]$ denote the block containing $n$. Conditioning on $\MOSM$,  $l$ and $r$ are integer random variable, such that $r-n$ and $n-1-l$ are stochastically dominated by $X$.
\end{restatable}



\vspace{-.2cm}
\begin{restatable}{lemma}{lemmameandomination}
    \label{lemma:meandomination}
    We have $\E[X] \leq \exp(\sum_{k\geq 1} k u_k) \sum_{k \geq 1} k^2 u_k$.
\end{restatable}


\subsection{Putting everything together}
\label{subsection:puttingeverythingtogether}
\begin{proof}[Proof of Theorem~\ref{thm:geom}] 
Without loss of generality, we may assume that $N = M \leq W$ and that each man is matched in the man-optimal stable matching $\mu_\M$: to see that, for each man $m$ we add a ``virtual'' woman $w$ as his least favorite acceptable partner, such that $m$ is the only acceptable partner of $w$. A man is single in the original instance if and only if he is matched to a ``virtual'' woman in the new instance. 

We start our analysis by conditioning on the random variable $\MOSM$ (see Definition~\ref{def:MOSM}). Algorithm~\ref{algo:MPDA} then computes $\mu_M$, which matches each woman to her worst stable partner. 
Up to \emph{relabeling} the women, we may also assume that for all $i\leq N$ we have $w_i:=\mu_{\M}(m_i)$.

Let $w_n$ be a woman who is married in $\mu_\M$. From there, we use Lemma~\ref{lemma:rankdifference}  to  bound the difference of rank between her worst and best stable partner by $x+r-l-1=x+(r-n)+(n-l-1)$. We bound the expected value of $x$ using Lemma \ref{lemma:delta}, and the expected values of both $r-n$ and $n-l-1$ using Lemmas~\ref{lemma:domination} 
and~\ref{lemma:meandomination}.
\end{proof}
\vspace{-.2cm}

\section{Unique stable partner: Proof of Theorem~\ref{thm:geomunif}}
\label{section:geomunif}
\newcommand{\OK}{\mathcal K}


\thmgeomunif*

\vspace{-.2cm} The proof first continues the analysis of blocks started in Section \ref{subsection:block}. When $u_k=\exp(-\Omega(k))$, it can be tightened with a mathematical analysis to prove (Corollary~\ref{corollary:failure}) that with high probability, no block size exceeds $O(\log n)$, and that in addition, in her preference list no woman switches the relative ordering of two men $m_i$ and $m_{i+\Omega(\log n)}$. The rest of the proof assumes that those properties hold.
The only remaining source of randomness comes from the preference lists of men.

The intuition is that it is hard for man $m_i$ to have another stable partner from his block. Because of the random uniform assumption on $m_i$'s preference list, between $w_i$ and the next person from his block, his list is likely to have some woman $w_j$ with $j\gg i$. Woman $w_j$ likes $m_i$ better than her own partner, because of the no-switching property, and $m_i$ likes her better than his putative second stable partner, so they form a blocking pair preventing $m_i$'s second stable partner. Transforming that intuition into a proof requires care because of the need to condition on several events.



\begin{restatable}{definition}{defOK}
\label{def:OK}
Let $C=O(1)$ be a constant to be defined later. Let $\OK$ denote the event that every block  has size at most $C \ln N$, and every woman prefers man $m_i$ to  man $m_{i+k}$
, whenever $k\geq C \ln N$.
\end{restatable}

\vspace{-.2cm}
The proofs of the following Lemmas can be found in Appendix~\ref{section:geomunif_proofs}.

\vspace{-.2cm}
\begin{restatable}{lemma}{lemmaexpblock}
    \label{lemma:expblock}
    Assume that women have preferences drawn from regular distributions such that
    $u_k = \exp(-\Omega(k))$.
    Then, the size of each man's block is a random variable with an exponential tail: 
    $$
    \forall i,\quad
    \P[\text{block containing }m_i\text{ has size }\geq k]=\exp(-\Omega(k)).
    $$ 
\end{restatable}


\vspace{-.2cm}
\begin{restatable}{corollary}{corollaryfailure}
  \label{corollary:failure}
  One can choose $C = \mathcal O(1)$ such that the probability of event $\OK$ is $\geq 1 - 1/N^2$.
\end{restatable}



\vspace{-.2cm}
\begin{restatable}{lemma}{lemmaonehusband}
  \label{lemma:onehusband}
  Fix $i\in [1,N]$.
  Conditioning on $\MOSM$ and on $\OK$,
  the probability that woman $w_i$ has more than
  one stable husband is at most $3 C \ln N/(N + C \ln N - i)$.
\end{restatable}


\begin{proof}[Proof of  Theorem~\ref{thm:geomunif}]
As in the previous proof, in our analysis we condition on event $\MOSM$ (see Definition~\ref{def:MOSM}), i.e. on (1) the man-optimal stable matching $\mu_\M$, (2) each man's ranking of the women he prefers to his partner in $\mu_\M$, and (3) each woman's ranking of the men who prefer her to their partner in $\mu_\M$. As before, a person who is not matched in $\mu_\M$ remains single in all stable matchings, hence, without loss of generality, we assume that $M = W = N$, and that  $w_i =\mu_\M(m_i)$ for all $1 \leq i \leq N$.

Let $Z$ denote the number of women with several stable partners. We show that in expectation $Z = \mathcal O(\ln^2 N)$, hence the fraction of persons with multiple stable partners converges to 0. We separate the analysis of $Z$ according to whether event $\OK$ holds. When $\OK$ does not hold, we bound that number by $N$, so by Corollary~\ref{corollary:failure}:
$
\E[Z]\leq (1/N^2)\times N+(1-1/N^2)\times \E(Z|\OK).
$
\par\noindent
Conditioning on $\MOSM$ and switching summations, we write:
$$\textstyle
\E(Z|\OK)
=\sum_{\MOSM} \P[\MOSM] \cdot \E(Z|\OK,\MOSM)
= \sum_i\sum_{\MOSM} \P[\MOSM]\cdot \P[w_i \text{ has several stable husbands}\;|\;\OK,\MOSM]
$$
By Lemma~\ref{lemma:onehusband}, we can write:
$
\P[w_i \text{ has several stable husbands}\;|\;\OK,\MOSM] \leq 3 C \ln N/(N + C \ln N - i).
$
Hence the expected number of women who have several stable partners is at most $1/N$ plus
\[\sum_{i=1}^N \frac{3 C \ln N}{N + C \ln N - i} =\sum_{i=0}^{N-1} \frac{ 3 C \ln N }{i + C \ln N}
\leq 3 C \ln N \int_{C \log N - 1}^{C\log N-1+N}\frac{\mathrm dt}{t} = 3C \ln N \ln\left(\frac{C\log N-1+N}{C \log N - 1}\right)\]
When $N$ is large enough, we can simplify this bound to $3 C \ln^2 N$.
\end{proof}

\section{Counting stable partners: Proof of Theorem~\ref{thm:popularity}}
\label{section:deferredacceptance}
In this section we analyze the classic algorithm from~\cite{knuth1990stable} that lists all stable husbands of a woman $w$, so as to study the number of stable husbands of a woman $w$ who has popularity preferences.  

In subsection~\ref{subsection:popularity}, we prove Lemma~\ref{lemma:popularity}, which will be used to compute the probability that a proposal is ``best so far''. In subsection \ref{subsection:popularityproof}, we prove Theorem~\ref{thm:popularity} as follows. First, we condition on~$\MOSM$ (see Definition~\ref{def:MOSM}) and compute the men-optimal stable matching, then we compute the sequence of proposals received by $w$ in Algorithm~\ref{algo:MPDAextended}, and finally we apply Lemma~\ref{lemma:popularity} to upper-bound the expected number of ``best so far'' proposals.

\vspace{-.2cm}
\subsection{Enumerating stable partners}
\label{subsection:enumerate}

Starting from the men-optimal stable matching (Algorithm~\ref{algo:MPDA}), Algorithm~\ref{algo:MPDAextended} continues the execution of the deferred acceptance procedure, in which $w$ rejects every proposal she receives. Stable husbands of $w$ are ``best so far'' proposals, that is men who proposed to $w$ and were preferred to all men who proposed before them.

\vspace{-.2cm}
\begin{theorem}[Adapted from \cite{knuth1990stable}]\label{lemma:MPDAextended} Algorithm~\ref{algo:MPDAextended} outputs $w^*$'s set of stable husbands.
\end{theorem}

\vspace{-.2cm}
\begin{algorithm}[h!]
\begin{algoindent}
\item\textbf{Input:}
Preferences of men $(\succ_m)_{m \in M}$ and women $(\succ_w)_{w \in \W}$. Fixed woman $w^* \in \W$.
\item\textbf{Initialization:}
Start by executing Algorithm \ref{algo:MPDA}, if $w^*$ is unmatched then stop.
\item\textbf{Phase 1: sequence of proposals}
\begin{algoindent}
    \item Let $m\gets \mu_\M(w^*)$ be the \emph{proposer}, let $S \gets [\mu_\M(w^*)]$ be the \textbf{sequence of proposals} to $w^*$.
    \item\textbf{While} the \emph{proposer} $m$ has not proposed to every woman he finds acceptable\textbf{, do}
    \begin{algoindent}
        \item $m$ proposes to his favorite woman $w$ he has not proposed to yet.
        \item \textbf{If} $w = w^*$\textbf{:}
        \\\quad append $m$ to the sequence $S$,
        \item \textbf{Else, if} $w$ has never been matched\textbf{:} \\\quad \textbf{break} the while loop,
        \item \textbf{Else, if } $m$ is $w$'s favorite acceptable man among all proposals she received\textbf{:}
        \\\quad $w$ rejects her previous husband $m'$, accepts $m$,   the \emph{proposer} becomes $m'$
    \end{algoindent}
\end{algoindent}
\item\textbf{Phase 2: stable husbands}
\begin{algoindent}
    \item \textbf{For each} proposal $m\in S$ made to $w^*$, in order of reception\textbf{:}
    \item \begin{algoindent}
        \item \textbf{If} $m$ is the best proposition $w$ received so far\textbf{:}
        \item\quad $m$ is a stable husband of $w^*$.
    \end{algoindent} 
\end{algoindent}
\item\textbf{Output:} Set of stable husbands of $w^*$.
\end{algoindent}
\def\figurename{Algorithm}
\caption{Extended Men Proposing Deferred Acceptance.}
\label{algo:MPDAextended}
\end{algorithm}
\vspace{-.6cm}

\subsection{Popularity preferences}
\label{subsection:popularity}

    Recall that in our model for probability preferences, the set of acceptable partners is deterministic. To build her preference list, $w$ samples without replacement from her set of acceptable men, first drawing her favourite partner, then her second favourite, and so on until her least favourite acceptable partner.
    
\begin{lemma}\label{lemma:popularity}
    Assume that a woman $w$ has popularity preferences defined by $\mathcal D_w$. Conditioning on a partial ranking of acceptable men $a_1 \succ_w \dots \succ_w a_p$, the probability that $w$ rank $m$ before $a_1$ is exactly
    \[\P[m\succ_{w}a_1\;|\;a_1\succ_w \dots \succ_w a_p] = \frac{\D_w(m)}{\D_w(m) + \sum_{i=1}^p \D_w(a_i)}\]
\end{lemma}

\begin{proof}
 One nice feature of popularity preferences is that to compute the probability that $a_1 \succ_w \dots \succ_w a_p$, one can ignore each time a man not in $\{a_1,\dots,a_p\}$ is drawn. We obtain
    \[\P[a_1 \succ_w \dots \succ_w a_p] = \prod_{i=1}^p \frac{\D_w(a_i)}{\sum_{j=i}^p \D_w(a_j)}.\]
    and similarly  for the probability that $m \succ_w a_1 \succ_w \dots \succ_w a_p$. Thus,
    \[\P[m\succ_{w}a_1\;|\;a_1\succ_w \dots \succ_w a_p \succ_w w] = \frac{\P[m\succ_{w}a_1\succ_w \dots \succ_w a_p \succ_w w]}{\P[a_1\succ_w \dots \succ_w a_p \succ_w w]} = \frac{\D_w(m)}{\D_w(m) + \sum_{i=1}^p \D_w(a_i)}\]
\end{proof}

\subsection{Proof of Theorem \ref{thm:popularity}}
\label{subsection:popularityproof}

\medbreak

\vspace{-.2cm}
\thmpopularity*

\vspace{-.5cm}
\begin{proof}
First, observe that $w$ is matched if and only if she receives a proposal in Algorithm~\ref{algo:MPDA}, which is independent from her ordering of acceptable men. By Theorem~\ref{lemma:MPDAextended}, Algorithm~\ref{algo:MPDAextended} outputs the stable husbands of $w$, so we analyze that algorithm, which starts by a call to Algorithm \ref{algo:MPDA}, which by Theorem~\ref{lemma:MPDA}  yields matching $\mu_\M$.

We know the preferences of everyone except $w$. We start the analysis by conditioning on the random variable $\MOSM$, i.e. on woman $w$'s ranking of the men who prefer her to $m_\M(w)$  (see Definition~\ref{def:MOSM}).  From here, observe that the execution of Algorithm~\ref{algo:MPDA}, and of the first phase of Algorithm~\ref{algo:MPDAextended} are deterministic. Let $x_0 = \mu_\M(w)$ be $w$'s worst stable husband, $K$ denote the number of proposals received by $w$ in Phase 1, and let $x_1, x_2, \dots, x_K$ denote the sequence of proposals received by $w$ during the first phase of Algorithm~\ref{algo:MPDAextended}.


Let  $p_\bot$ denote the sum of popularities of proposals received by $w$ during the initial call to Algorithm~\ref{algo:MPDA}, including $\mu_\M(w)$, and let  $p_i = \D_w(x_i)$ for all $0 \leq i \leq K$. By linearity of expectations, and then using Lemma \ref{lemma:popularity}, 
\vspace{-.2cm}
\begin{eqnarray}
\E[\hbox{Nb of stable husbands of }w\,|\,\MOSM] 
&=& 1+ \textstyle\sum_{i=1}^K \P[\hbox{proposal $x_i$ is accepted by $w$} \;|\;\MOSM] \\
& = & 1 +\textstyle\sum_{i=1}^K p_i/(p_\bot+p_1+\dots+p_i).
\label{eq:1}
\end{eqnarray}
\vspace{-.2cm}
We simplify the right-hand side with a sum-integral comparison: 
\begin{equation}\label{eq:sumintegral}
\sum_{i=1}^K \frac{p_i}{p_\bot + p_1 +  \dots + p_i} \leq \sum_{i=1}^K\int_{p_\bot + p_1 +\dots+p_{i-1}}^{p_\bot+p_1+\dots+p_{i}} \frac{\mathrm dt}{t} = \ln(p_\bot + p_1 +\dots+p_K) - \ln p_\bot .
\end{equation}
We use the convexity of $t \mapsto t \ln t$ and Jensen's inequality:
\begin{equation}\label{eq:jensen}
\ln(p_\bot + p_1 +\dots+p_K) \leq\ln (K+1)+ \frac{p_\bot\ln p_\bot + p_1 \ln p_1 + \dots + p_K \ln p_K}{p_\bot + p_1 + \dots + p_K}.
\end{equation}
We now focus on the right-hand side of the equation in Theorem \ref{lemma:popularity}. 
By definition of $\mu_\W$, man $\mu_\W(w)$ is the overall best proposition received by $w$.
It is $x_0$ with probability proportional to $p_\bot$ and it is $x_i$ ($1 \leq i \leq K$) with probability
proportional to $p_i$, thus
\begin{equation}\label{eq:righthand}
\E\left[\ln(\D_w(\mu_\W(w)))\,|\,\MOSM\right]=
\frac{p_\bot\ln p_0 + p_1 \ln p_1 + \dots + p_K \ln p_K}{p_\bot + p_1 + \dots + p_K}.
\end{equation} 
Combining equations (\ref{eq:1}), (\ref{eq:sumintegral}), (\ref{eq:jensen}) and (\ref{eq:righthand}), we obtain 
$$\E[\hbox{Nb of stable husbands of }w] \leq 
1 + \ln(K+1) + \E[\ln(\D(\mu_\W(w^*)))] - \frac{p_\bot \ln p_0 - p_\bot \ln p_\bot}{p_\bot + p_1 + \dots + p_K} - \ln p_\bot,
$$
where all expectations are conditioned on $\MOSM$. 
Since $p_\bot \geq p_0$, we can write
$$
 \frac{p_\bot \ln p_0 - p_\bot \ln p_\bot}{p_\bot + p_1 + \dots + p_K} + \ln p_\bot = 
  \frac{p_\bot \ln p_0 +(p_1+\ldots +p_K) \ln p_\bot }{p_\bot + p_1 + \dots + p_K}\geq\ln p_0 = \ln \D_w(\mu_\M(w)).
$$
To conclude the proof, observe that $K+1\leq d_w$ and take expectations over $\MOSM$.
\end{proof}
\vspace{-.2cm}



\section{Counting stable partners: Proof of Theorem~\ref{theorem:womanbounded}}
This section provides
a brief sketch of proof of Theorem~\ref{theorem:womanbounded}.
The full proof can be found in the appendix.
This theorem provides a bound on the expected number of stable husbands of a woman when both men and women have popularity preferences.
The bound depends on how uniform the preferences of men are (parameter $R_\M$)
and how similar the preferences of women are (parameter $Q_\W\ref{theorem:womanbounded}$).Both parameters are formally defined in the introduction.

\womanboundedtheorem*

The proof of Theorem~\ref{theorem:womanbounded} relies  on the computation of the set of stable husbands of $w$ by Algorithm~\ref{algo:MPDAextended}. 
The execution of this algorithm is considered as a stochastic process, where preferences of men are progressively revealed.

Say for simplicity that $R_\M=1$ (men have uniform preferences over women) and $Q_\W=1$ (women agree on the popularities of men).
We also rule out some events which occur with a negligible probability. Set $T=N^5$. We assume that no woman may prefer a man $h$ to another man $h'$ if  
$h'$ is more than $T$ times more popular then $h$. Such an event occurs with probability at most $\frac{1}{N^2}$.

Consider the initial husband $m$ of $w$ in the man-proposing matching $\mu_\M$.
Then $w$ cannot have any husband $T$ times less popular than $h$. What about husbands more popular than $h$?

We order men by increasing popularity $m_1,m_2,\ldots , m_N$. Let $i$ such that $m=m_{i}$. We partition the set of men in sets of exponentially increasing size,
starting with all men less or equally popular than $m$.
Let $F_0=[1,m_i]$ ,$F_1=(m_i,2*m_i]$,
$F_2=(2*m_i,4*m_i]$,..., $F_j=(2^k*i,N]$,
where $2^j*i < N \leq 2^{j+1}*i$.

Set 
$L=\frac{4\ln(N)}{\ln\left(1 + \frac{1}{R_\M} \right)}$
and $K=(T\cdot Q_\W)^L$.
For every set $F_\ell,\ell\in 1\ldots j$,
we say that \emph{there is a huge popularity gap} in  $F_\ell$ if the popularity ratio for $w$ in this interval is $\geq K$, i.e. if $\D_w(m_{2^\ell*m_i}) \leq K \D_w(m_{\min(N,2^{\ell+1}*m_i})$.

In case there is no huge popularity gap in any of the $F_\ell,\ell\in 1\ldots j$ then the upper bound follows easily since $j\leq \ln_2(N)$.

Otherwise we select the smallest $\ell$ for which there is huge popularity gap in $F_\ell$.
Denote $E_0=F_0\cup F_1\cup\ldots F_{\ell-1}$
and $E_1=F_\ell$.
Moreover, denote $E_1'\subseteq E_1$ the set of men in $E_1$
whose popularity is above $T$ times the least popular men in $E_1$ and below $T$ times the most popular man in $E_1$.

Assume that along the progressive revelation 
of the preferences of men, the process reaches a man in $E_1'$. Since men have uniform preferences over all women and since $|E_0|\geq |E_1|$
there is probability $\geq \frac{1}{2}$
that the process falls back down to $E_0$ at the next step.

Thus from $E_1'$ the process is strongly attracted down to $E_0$ and since it lasts for at most $N^2$ steps, with high probability the process will not perform more than $\mathcal{O}(\ln(N))$ consecutive steps in $E_1'$.
Since the popularity increases by a factor of at most $T$ at every step,
the top popularity stays below $T^{\mathcal{O}(\ln(N))}$
times the popularity of the most popular man in $E_0$. That leads to the bound in the theorem.

From the upper bound on the popularities,
one extracts an upper bound on the expected number of husbands (Corollary~\ref{cor:womanbounded}) thanks to the following lemma.

\begin{restatable}{lemma}{womanboundedlemma}
\label{lem:prel}
Let $w$ be a woman married to a man
$\mu_\M(w))$ in the man-proposing stable matching.
Denote $L_w$ the list of proposals received
by $w$ during the enumeration of stable husbands computed by Algorithm~\ref{algo:MPDAextended}.
Then
\[
\mathbb{E}[\text{number of stable husbands of $w$} \mid \mu_\M]
\leq
\ln(|L_w|) + \ln\left(\max_{m\in L_w} \frac{\D_w(m)}{\D_w(\mu_\M(w))}\right)\enspace.
\]
\end{restatable}


The bound on the number of stable husbands
given in Theorem~\ref{theorem:womanbounded}
is polylog in $N$ only if the preferences of men on women are close to uniform ($R_M$ should be polylog in $N$). We conjecture that this hypothesis is not necessary.

\newpage
\bibliographystyle{alpha}
\bibliography{biblio}

\appendix

\newpage
\section{Proofs from subsection~\ref{subsection:block}}
\label{section:geom_proofs}
\lemmadomination*

\begin{proof}
    Conditioning on $\MOSM$, we know the men-optimal stable matching $\mu_\M$, and each woman's ranking of the men who prefer her to their partner in $\mu_\M$. Using notations from Lemma~\ref{lemma:regular}, let $J_i = \max\{j\geq i\;|\;m_j \succeq_{w_i} m_i\}$, for all $1 \leq i \leq M$.

    From Lemma~\ref{lemma:correctnessalgorithm2}, the right separator $r$ is computed with a while loop. Let $r_0 = n$ be the initial value of $r$.
    To decide whether $r_0$ is a separator, we look at $w_n$'s preference list. Let $r_1 = J_n$ be the maximum $j \geq n$ such that $w_n$ prefer $m_j$ to $m_n$. If $r_1 = r_0$, $w_n$ prefers $m_n$ to all men $m_j$ with $j > n$, and $r_0$ is a prefix separator. Otherwise, no prefix separator can exist between $r_0$ and $r_1$. Using Lemma~\ref{lemma:regular}, $r_1-r_0$ is stochastically dominated by $\Delta_0$.
    \[\forall \delta > 0,\quad \P[r_1-r_0<\delta\;|\;\MOSM] = \P[J_n<n+\delta\;|\;\MOSM] \geq \exp(\textstyle-\sum_{k \geq \delta} u_k) \geq \P[\Delta_0 < \delta]\]
    For all $t > 0$, we proceed by induction.
    To decide whether $r_{t}$ is a separator, we look at the preference lists of $w_{1+r_{t-1}}, \dots, w_{r_t}$. Let $r_{t+1} = \max\{J_{1+r_{t-1}},\dots,J_{r_t}\}$ be the maximum $j \geq r_t$ such that a woman $w_i$ prefer $m_j$ to $m_i$, with $r_{t-1} < i \leq r_t$. If $r_{t+1} = r_t$, then $r_t$ is a prefix separator. Otherwise, no prefix separator can exist between $r_t$ and $r_{t+1}$. We show that $\Delta_t$ stochastically dominates $r_{t+1}-r_t$. 
    \begin{align*}
        \forall \delta>0,\quad \P[r_{t+1}-r_t<\delta\;|\;\MOSM,J_{n}, \dots, J_{r_{t-1}}] 
        &=\prod_{i=1+r_{t-1}}^{r_t}\!\!\P[J_i<r_t+\delta\;|\;\MOSM,J_{n}, \dots, J_{r_{t-1}}]\\
        &\geq\prod_{i=1+r_{t-1}}^{r_t}\!\!\exp(\textstyle- \sum_{k\geq r_t+\delta-i} u_k)\quad(\text{Lemma}\;\ref{lemma:regular})\\
        &\geq\exp(\textstyle- \sum_{k\geq \delta} k u_k) = \P[\Delta_t < \delta]
    \end{align*}
    Summing up to $t$ such that $r_{t+1}=r_t$ proves that $X$ stochastically dominates $r-n$.


    We now prove that $X$ stochastically dominates $n-1-l$. From Lemma~\ref{lemma:correctnessalgorithm2}, the left separator $l$ is computed with a while loop, and let $l_0 = n-1$ be its initial value. To decide whether $l_0$ is a prefix separator, we need to know if a woman $w_i$ prefers a man $m_j$ to her husband $m_i$, with $i \leq l_0 < j$. More formally, $l_0$ is a prefix separator if and only if $J_i \leq l_0$ for all $i \leq l_0$. Defining $l_1 = \min\{i \leq l_t+1\;|\;J_i > l_t\} - 1$, $l_1 = l_0$ if and only if $l_0$ is a separator. Using Lemma~\ref{lemma:regular}, $l_1-l_0$ is stochastically dominated by $\Delta_0$.
    \[\forall \delta > 0,\quad \P[l_0-l_1<\delta\;|\;\MOSM] = \P[J_1, \dots, J_{l_0-\delta+1} \leq l_0\;|\;\MOSM] \geq \exp(\textstyle-\sum_{k \geq \delta} u_k) \geq \P[\Delta_0 < \delta]\]
    For all $t > 0$, we proceed by induction and let $l_{t+1} = \min\{i \leq l_t+1\;|\;J_i > l_t\} - 1$. More precisely, $l_{t+1}+1$ is the minimum $i \leq l_t+1$ such that $w_i$ prefer a man $m_j$ to her husband $m_i$ with $j > l_t$. If $l_{t+1}=l_{t}$, then $l_t$ is a prefix separator, and the process stop here. Otherwise, no prefix separator can exist between $l_{t+1}$ and $l_{t}$. A crucial property is that for all $i \leq l_{t+1}$, the best man in $w_i$'s partial list is still $m_i$, hence Lemma~\ref{lemma:regular} will still be applicable the next step.
    \begin{align*}
        \forall \delta>0,\quad \P\left[l_{t}-l_{t+1}<\delta\;\middle|\;\begin{array}{l}
        J_1,\dots, J_{l_t} \leq l_{t-1}\\
        \MOSM,l_0, \dots, l_t 
        \end{array} \right]
        = &\prod_{i=1}^{l_t-\delta+1}
        P[J_i \leq l_t\;|\; J_i \leq l_{t-1}, \MOSM]\\
        \geq & \prod_{i=1}^{l_t-\delta+1}
        \exp(\textstyle- \sum_{k\geq l_t+1-i} u_k)~~~~~(Lemma~\ref{lemma:regular})\\
        \geq&\exp(\textstyle- \sum_{k\geq \delta} k u_k) = \P[\Delta_t < \delta]
    \end{align*}
    Summing up to $t$ such that $l_{t+1}=l_t$ proves that $X$ stochastically dominates $n-1-l$.
\end{proof}

\lemmameandomination*
\begin{proof}[Proof of Lemma~\ref{lemma:meandomination}]
    From Wald's equation, $\E[X] = \E[T] \cdot \E[\Delta_0]$.
    The random variable $T$ is geometrically distributed, with a success parameter $\P[\Delta_0 = 0]$, hence $E[T] = 1/\P[\Delta_0 = 0]$. Because $\Delta_0$ only takes non-negative integer values, we can compute its expectation with a sum.
    \[\E[\Delta_0] = \sum_{\delta \geq 0} \P[\Delta_0 > \delta] = \sum_{\delta \geq 0} 1 - \exp\left(\textstyle-\sum_{k>\delta}k u_k\right) \leq \sum_{\delta \geq 0} \sum_{k>\delta}k u_k = \sum_{k \geq 1} k^2 u_k\]
\end{proof}

\section{Proofs from section~\ref{section:geomunif}}
\label{section:geomunif_proofs}
\defOK*

\lemmaexpblock*
\begin{proof}
    Recall that by Lemma~\ref{lemma:correctnessalgorithm2} blocks can be computed using Algorithm~\ref{algo:separators}. Let $(l,r]$ be the block of man $m_n$. By
    Lemma~\ref{lemma:domination},
    both $r-n$ and $n-1-l$ are stochastically dominated by the random variable~$X$. If $X$ has an exponential tail, one can conclude the proof using the union bound.
    
    Let $G_X(z) = \E[z^X]$ be the probability generating function of $X$, which is defined at least for all real $z$ such that $|z| < 1$. In addition if $G_X(1+\varepsilon)$ is finite for some $\varepsilon > 0$, then Markov's inequality gives
    \[\forall k \geq 0,\quad \P[X \geq k] = \P[(1+\varepsilon)^X \geq (1+\varepsilon)^k] \leq G_X(1+\varepsilon) \exp(-k\ln(1+\varepsilon)) = \exp(-\Omega(k)).\]
    Computing $G_X$ using Definition~\ref{def:domination}, and conditioning on the value of $T$.
\[G_{X}(z) = \E[z^{X}] = \sum_{t=0}^{+\infty}\P[T = t] \cdot \E\left[z^{\sum_{i=0}^{t-1}\Delta_i}\;\middle|\;\forall i \in [0,t-1],\;\Delta_i > 0 \right]\]
Using the fact that all $ \Delta_i$'s are \textit{i.i.d.} we can simplify the expectation of the product.
\[G_{X}(z) = \cdot\sum_{t=0}^{+\infty}\P[T = t] \cdot \E\left[z^{\Delta_0}\;\middle|\;\Delta_0 > 0\right]^t = G_{ T}\left(\E\left[z^{\Delta_0}\;\middle|\;\Delta_0 > 0\right]\right)\]
The conditional expectation can be expressed as follows.
\[G_{\Delta_0}(z) = \E\left[z^{\Delta_0}\right] = \P[\Delta_0 > 0] \cdot \E\left[z^{\Delta_0}\;\middle|\;\Delta_0 > 0\right] + \P[\Delta_0 = 0]\]
\[\E\left[z^{\Delta_0}\;\middle|\;\Delta_0 > 0\right] = \frac{G_{\Delta_0}(z) - \P[\Delta_0 = 0]}{\P[\Delta_0 > 0]}\]
Now let us compute the generating function of $T$.
\[G_{T}(z) = \E[z^{T}] = \sum_{t=0}^{+\infty} z^t \cdot \P[T = t] = \sum_{k=0}^{+\infty} z^t \cdot \P[\Delta_0 > 0]^t \cdot \P[ \Delta_0 = 0] = \frac{\P[\Delta_0 = 0]}{1 - z \cdot \P[\Delta_0 > 0]}\]
Combining the three previous equations we obtain
\[G_{X}(z) = \frac{\P[\Delta_0 = 0]}{1 + \P[\Delta_0 = 0] - G_{\Delta_0}(z) }\]
Because of the assumption on women's preference distributions, we have $u_k = \exp(-\Omega(k))$. Hence,
\[\forall \delta \geq 1,\quad \P[\Delta_0=\delta] = \P[\Delta_0 < \delta+1]-\P[\Delta_0<\delta] = \textstyle \exp(-\sum_{k>\delta} ku_k) (1 - \exp(-\delta u_\delta)) \leq \delta u_\delta = \exp(-\Omega(\delta))\]
Thus, the convergence radius of $G_{\Delta_0}$ is strictly greater than 1. Because $G_{\Delta_0}$ is a probability generating function, it is continuous, strictly increasing, and $G_{\Delta_0}(1) = 1$. Therefore, there exists $\varepsilon > 0$ such that $G_{\Delta_0}(1+\varepsilon) < 1+\P[\Delta_0 = 0]$, which concludes the proof.
\end{proof}

\corollaryfailure*

\begin{proof}
  For the first case of failure,
  recall from Lemma~\ref{lemma:expblock} that the size of a block has an exponential tail.
  Thus we can choose $C$ such that the probability of a given block has a size greater than
  $C \log N$ is at most $1/(2N^3)$. There are at most $N$ blocks, using the union bound the
  probability that at least one has a size exceeding $C \log N$ is at most $1/(2N^2)$.
  
  For the second case of failure, notice that the probability for a woman to prefer a
  man $m_j$ to a another man $m_i$ with $j > i + C \ln N\leq j$ is at most $u_{j-i} = e^{-\Omega(j-i)} = N^{-C\Omega(1)}$.
  Thus we can choose $C$ such that the probability of this happening is smaller
  than $1/(2N^5)$.
  Using the union bound over all triples of woman/$m_i$/$m_j$, the probability
  of a failure is at most $1/(2N^2)$.
  
  Choosing $C$ maximal between the two values, and using the union
  bound over the two possible cases of failure, the probability that $\OK$
  does not holds is at most $1/N^2$.
\end{proof}

\lemmaonehusband*

\begin{figure}[h]
\vspace{-.2cm}
  \begin{center}
  \scalebox{.8}{\begin{tikzpicture}
    \def\xscale{1.5}
    \fill[red!20] (-.25*\xscale,3) -- (1.75*\xscale,3) -- (3.5*\xscale,1);
    \fill[orange!20] (1.75*\xscale,3) -- (7.75*\xscale,3) -- (3.5*\xscale,1);
    \fill[green!20] (7.75*\xscale,3) -- (9.75*\xscale,3) -- (3.5*\xscale,1);
    \fill[red!40] (-0.25*\xscale,2.8) rectangle (1.75*\xscale,3.2);
    \fill[orange!40] (1.75*\xscale,2.8) rectangle (7.75*\xscale,3.2);
    \fill[green!40] (7.75*\xscale,2.8) rectangle (9.75*\xscale,3.2);
    \fill[green!40] (9.80*\xscale,2.8) rectangle (9.90*\xscale,3.2);
    \fill[green!40] (9.95*\xscale,2.8) rectangle (10*\xscale,3.2);
    \draw[->,>=stealth] (0*\xscale,0) -- (10*\xscale,0);
    \draw (-.5*\xscale,3) node[anchor=east] {Women};
    \draw (-.5*\xscale,1) node[anchor=east] {Men};
    \foreach \x in {0,.5,...,9.5} {
      \draw (\x*\xscale,3) circle (1.5pt);
      \draw (\x*\xscale,1) circle (1.5pt);
      \draw (\x*\xscale,-.1) -- (\x*\xscale,.1);
      \draw (\x*\xscale, 1) -- (\x*\xscale, 3);
    }
    \draw[dashed,->,>=stealth, line width=2pt,red] (3.5*\xscale,1) -- (1.0*\xscale,3);
    \draw[->,>=stealth, line width=2pt,green!60!black] (3.5*\xscale,1) -- (8.0*\xscale,3);
    \draw[densely dotted,>=stealth] (6.5*\xscale,1) -- (8.0*\xscale,3);
    \draw[densely dotted,>=stealth] (3.5*\xscale,1) -- (4.0*\xscale,3);
    \fill[red] (2.25*\xscale,2) circle (.5cm);
    \fill[white] (2*\xscale,1.9) rectangle (2.5*\xscale,2.1);
    \draw (3.5*\xscale, .5) node {$m_i$};
    \draw (3.5*\xscale, 3.5) node {$w_i$};
    \draw (8.0*\xscale, 3.5) node {$w^*$};
    \draw (10.2*\xscale, 2) node {$\dots$};
    \draw[dashed] (-.25*\xscale,-.2) -- (-.25*\xscale,3.7);
    \draw[dashed] (1.25*\xscale,-.2) -- (1.25*\xscale,3.7);
    \draw[dashed] (2.75*\xscale,-.2) -- (2.75*\xscale,3.7);
    \draw[dashed] (4.25*\xscale,-.2) -- (4.25*\xscale,3.7);
    \draw[dashed] (6.25*\xscale,-.2) -- (6.25*\xscale,3.7);
    \draw[dashed] (8.25*\xscale,-.2) -- (8.25*\xscale,3.7);
    \draw[dashed] (9.75*\xscale,-.2) -- (9.75*\xscale,3.7);
    \draw[arrows={<[length=.1cm]}-{>[length=.1cm]}] (1.5*\xscale,-.5) -- (3.5*\xscale,-.5);
    \draw[arrows={<[length=.1cm]}-{>[length=.1cm]}] (3.5*\xscale,-.5) -- (5.5*\xscale,-.5);
    \draw[arrows={<[length=.1cm]}-{>[length=.1cm]}] (5.5*\xscale,-.5) -- (7.5*\xscale,-.5);
    \draw (2.5*\xscale, -.8) node {$C\ln N$};
    \draw (4.5*\xscale, -.8) node {$C\ln N$};
    \draw (6.5*\xscale, -.8) node {$C\ln N$};
  \end{tikzpicture}}
  \end{center}
  \vspace{-.3cm}
  \caption{Proof of Lemma~\ref{lemma:onehusband},
  probability that $w_i$ has several stable husbands $\leq$ ratio $\frac{\colorbox{orange!30}{\dots}}{\colorbox{orange!30}{\dots} \;+\; \colorbox{green!30}{\dots}}$}
  \label{figure:uniformproposals}
  \vspace{-.5cm}
\end{figure}
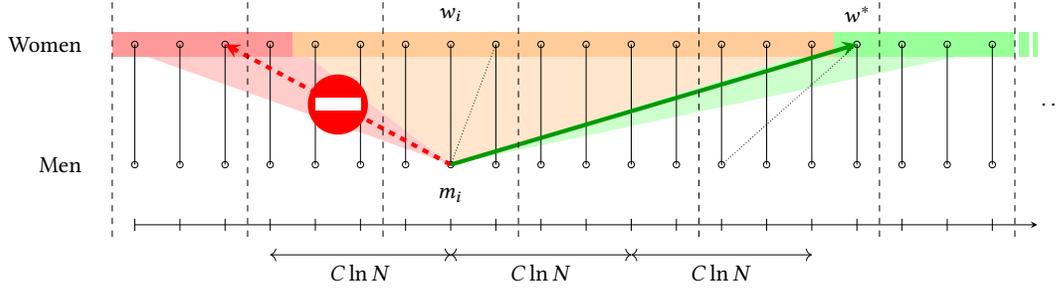

\begin{proof}  
Say that a woman $w_k$ with $k\neq i$ and to whom $m_i$ prefers $w_i$ is ``red''  if $k \leq i - C \ln N$, ``yellow'' if $i - C \ln N < k \leq i + 2C \ln N$, and ``green'' if $i + 2C \ln N < k$.  
Let $R$, $Y$ and $G$ be the sets of red, yellow and green women. Women who are not colored are ranked by $m_i$ better than $w_i$, his best stable partner, so they cannot be stable partners of $m_i$.

If woman $w_i$ has another stable partner besides $m_i$, then man $m_i$ also has at least one other stable partner. 
Because of $\OK$, every red woman $w_k$ prefers $m_k$ to $m_i$. Since $m_k$ is her worst stable partner, there is no stable marriage in which $m_i$ is paired with $w_k$. Thus all stable partners of $m_i$ must be among  $ Y \cup G$.

Let $w^*$ be $m_i$'s favorite woman among  $ Y \cup G$.  We will argue that if $w^*\in G$ then $w_i$ is $m_i$'s unique stable partner.  Assume, for a contradiction, that $m_i$ has another stable partner $w$ besides $w_i$, and consider that stable matching $\mu$. By  Lemma~\ref{lemma:separator}, $w$ must belong to $i$'s block. By $\OK$ and since $w^*\in G$, $w^*$ is in a different block, so $w\neq w^*$. Consider the pair $(m_i,w^*)$. By definition of $w^*$, man $m_i$ prefers $w^*$ to $w$. By $\OK$ and definition of $G$, $w^*$ prefers $m_i$ to the man of her block to whom she is married in $\mu$. So $(m_i,w^*)$ is a blocking pair, contradicting stability of $\mu$. This proves
$$ 
\P[w_i \text{ has more than 1 stable partner}]\leq \P[w^*\in Y].
$$
Recall that $m_i$'s preferences are uniformly random. Once we condition on $\MOSM$, the preferences of $m_i$ are still uniform over all the women to whom $m_i$ prefers $w_i$. Event $\OK$ only depends on the women's preference lists, so conditioning on $\OK$ does not change that, thus
$$ \P[w^*\in Y] =  \frac{|Y|}{|Y|+|G|} \leq \frac{3C\ln N}{3C \ln N + |G|}.$$
Finally, we argue that all women $w_j$ with $j>i + 2C \ln N$ are in $G$. Consider a woman $w_j$ with $j>i + 2C \ln N$. Conditioning on $\OK$,  $w_j$ prefers $m_i$ to her partner $w_j$ in $\mu_\M$, so by stability of $\mu_\M$, man $m_i$ prefers $w_i$ to $w_j$,  so $w_j\in G$. Thus, conditioning on $\MOSM$ and $\OK$,
we have $|G| = N-i-2C\ln N$.
\end{proof}

\section{Tightness of Theorem~\ref{thm:popularity}}
\label{section:onewomanlb}
In \cite{knuth1990stable}, Knuth, Motwanni and Pittel prove that when all persons have complete uniform preference lists, the upper-bound of Theorem \ref{thm:popularity}  is essentially tight. Here, we give another example showing that the upper bound from Theorem~\ref{thm:popularity} is also tight when $w$ has complete popularity preferences $\D_w: m_i \mapsto \lambda^i$ with parameter $\lambda = 0.99$.

In the upper bound from Theorem~\ref{thm:popularity}, the ratio of popularity is at most $\lambda^{1-M}$, hence its logarithm is at most $(1-M)\ln\lambda$.
When $\lambda = 0.99$, we get $(1-M)\ln\lambda \approx (1-\lambda) \cdot M$, hence Theorem~\ref{thm:popularity} states that at most $\approx 1\%$ of the men are stable husbands of $w$. Lemma~\ref{lemma:popularitylb} proves that there exists an instance such that this $1\%$ upper bound is asymptotically tight.

\begin{lemma}
\label{lemma:popularitylb}
Let $w$ a woman having complete popularity preferences $\D_w: m_i \mapsto \lambda^i$ with $0 < \lambda < 1$. One can choose the preference lists of the other persons such that:
\[\E[\text{Number of stable husbands of }w] > (1-\lambda) \cdot M\]
\end{lemma}

\begin{proof}
Take a community with $N$ men and $N$ women. We adapt a folklore instance where each man-woman pair is stable. We replace the preference list of woman $w_1$ by a complete popularity preference list defined by $\D_w: m_i \mapsto \lambda^i$ with $0 < \lambda < 1$, which tends to be to similar with the original preference list $m_1 \succ m_2 \succ \dots \succ m_N$.
\[
\begin{array}{cc|ccccccccccc}
    m_1 && w_2 &\succ_{m_1}& w_3 &\succ_{m_1}& \dots &\succ_{m_1} & w_{N-1} & \succ_{m_1} & w_N &\succ_{m_1}& w_1 \\
    m_2 && w_3 &\succ_{m_2}& \dots &\succ_{m_2}&  w_{N-1} & \succ_{m_2} & w_N &\succ_{m_2}& w_1 &\succ_{m_2}& w_2 \\
    \vdots && \vdots & & & & & & & & & & \vdots \\
    m_{N-1} && w_N &\succ_{m_{N-1}}& w_1& \succ_{m_{N-1}}& w_2 &\succ_{m_{N-1}}& w_3 &\succ_{m_{N-1}}& \dots & \succ_{m_{N-1}} & w_{N-1}\\
    m_N && w_1 &\succ_{m_N}& w_2 &\succ_{m_N}& w_3 &\succ_{m_N}& \dots & \succ_{m_N} & w_{N-1} & \succ_{m_N}& w_N \\
    \\
    w_1 && \multicolumn{11}{l}{\fbox{complete popularity preferences $\D_w: m_i \mapsto \lambda^i$, for some  $0 < \lambda < 1$.}} \\
    w_2 && m_2 &\succ_{w_2}& \dots &\succ_{w_2}&  m_{N-2} & \succ_{w_2} & m_{N-1} &\succ_{w_2}& m_N &\succ_{w_2}& m_1 \\
    \vdots && \vdots & & & & & & & & & & \vdots \\
    w_{N-1} && m_{N-1} &\succ_{w_{N-1}}& m_N& \succ_{w_{N-1}}& m_1 &\succ_{w_{N-1}}& m_2 &\succ_{w_{N-1}}& \dots & \succ_{m_{N-1}} & m_{N-2} \\
    w_N && m_N &\succ_{w_N}& m_1 &\succ_{w_N}& m_2 &\succ_{w_N}& \dots & \succ_{w_N} & m_{N-2} & \succ_{w_N}& m_{N-1}
\end{array}
\]
\smallbreak
During the execution of Algorithm~\ref{algo:MPDA}, each man proposes to his favorite woman, and each woman receives exactly one proposition. Thus in the man-optimal stable matching, we have $\mu_\M(w_1) = m_N$, $\mu_\M(w_2) = m_1$, \dots, $\mu_\M(w_N) = m_{N-1}$.
\medbreak
Then, during the execution of Algorithm~\ref{algo:MPDAextended}, woman $w^* = w_1$ will receive propositions from $m_{N-1}$, $\dots$, $m_2$, $m_1$, exactly in that order. The proposal from man $m_i$ will be the ``best so far'' with probability $\lambda^i / \sum_{j=i}^N \lambda^j$. Hence, $w^*$'s expected number of stable husbands is exactly
\[
\E[\text{Number of stable husbands of }w] = \sum_{i=1}^N \frac{\lambda^i}{\sum_{j=i}^N \lambda^j} = \sum_{i=1}^N \frac{1-\lambda}{1-\lambda^{N-i+1}} > (1-\lambda) \cdot N
\]

\end{proof}

\section{Telescopic sums: Proof of Corollary \ref{thm:telescopic}}
\label{section:telescopic}
In this section, we repeatedly use Theorem~\ref{thm:popularity} to prove Corollary~\ref{thm:telescopic}. In both settings \textit{(1)} and \textit{(2)}, we express the number of stable pairs as a sum of number of stable partners, in which some terms cancel out nicely.

\thmtelescopic*

\medbreak
In subsection~\ref{section:idpop} and~\ref{section:sympop} we prove Theorems~\ref{thm:idpop} and~\ref{thm:sympop} which are stronger versions of settings \textit{(1)} and~\textit{(2)}, where we relax popularity constraints.

\subsection{Intrinsic popularities}
\label{section:idpop}

In this subsection, all women have popularity preferences. We say that men have \emph{intrinsic} popularities when all women agree on the popularity of each man they find acceptable. To measure the extent to which women agree on the popularity of men, we introduce a parameter $r_m \geq 1$ which is the ratio between the highest and the lowest popularity given to man $m$ by some woman who finds him acceptable. Lower values of $r_m$'s mean more correlations between the preferences of women. When $r_m = 1$, all women agree on the intrinsic popularity of man $m$.

\begin{theorem}
\label{thm:idpop}
Assume that each woman $w$ has popularity preferences defined by $\D_w$ over a set of $d_w \geq 1$ acceptable men. For each man $m$, we define the ratio $r_m$ between the highest and the lowest popularity given to $m$ by some woman who finds him acceptable. Then:
\[\E[\text{Number of stable pairs}] \leq N + \sum_{w \in \W} \ln d_w + \sum_{m \in \M} \ln r_m\]
When all women agree on the popularity of men, the number of stable pairs is at most $N + N \ln N$.
\end{theorem}

\begin{proof}
Theorem~\ref{thm:popularity} is valid for each woman $w\in\W$. Indeed, the case where all the other women have popularity preferences is actually a linear combination of cases where those women have deterministic preferences. However, one needs to deal with the assumption that $w$ receives at least one acceptable proposal during Algorithm~\ref{algo:MPDA}.
\medbreak
Each person is either matched in all stable matchings or single in all stable matchings. For each person $p$, define $X_p$ the event where $p$ is matched. For every woman $w$, event $X_w$ is true if and only if $w$ is receives at least one acceptable proposal during Algorithm~\ref{algo:MPDA}, which does not depend on the preference list of $w$. If we write $Y_w$ the number of stable husbands of $w$, we have:
\[\forall w\in \W,\quad\E[Y_w\;|\;X_w] \leq 1 + \ln d_w + \E\left[\ln\frac{\D_w(\mu_\W(w))}{\D_w(\mu_\M(w))}\;|\;X_w\right] \]
We write $Y = \sum_{w \in \W} Y_w$ the total number of stable pairs. Using linearity of expectation we obtain: 
\begin{align}
\E[Y] = \sum_{w\in \W} \E[Y_w\;|\;X_w]\cdot\P[X_w] &\leq \sum_{w\in\W}(1 + \ln d_w)\cdot \P[X_w] \nonumber \\
&+ \sum_{w\in\W} \E[\ln(\D_w(\mu_\W(w)))\;|\;X_w]\cdot\P[X_w] \label{thm:idpop:1}\\
&- \sum_{w\in\W} \E[\ln(\D_w(\mu_\M(w)))\;|\;X_w]\cdot \P[X_w] \label{thm:idpop:2}
\end{align}
The two sums (\ref{thm:idpop:1}) and (\ref{thm:idpop:2}) can be rewritten as sums over men: for every stable matching $\mu$ we have
\[\sum_{w\in\W} \E[\ln(\D_w(\mu(w)))\;|\;X_w]\cdot\P[X_w] = 
\sum_{m\in\M} \E[\ln(\D_{\mu(m)}(m))\;|\;X_m]\cdot\P[X_m]\]
 Using once again linearity of expectation:
\[\E[Y] \leq \sum_{w\in\W}(1 + \ln d_w)\cdot \P[X_w] + \sum_{m\in\M} \E\bigg[\ln \frac{\D_{\mu_\W(m)}(m)}{\D_{\mu_\M(m)}(m)}\;|\;X_m\bigg]\cdot \P[X_m]\]
To conclude the proof, observe that the sum $\sum_w \P[X_w]$ is at most $N=\min(M,W)$, and that the ratio $\D_{\mu_\W(m)}(m)/\D_{\mu_\M(m)}(m)$ is at most $r_m$.
\end{proof}





As long as the ratios $r_m$ are polynomial in $N$, the expected number of stable pairs is $O(N \ln N)$.

\subsection{Symmetric popularities}
\label{section:sympop}

In this subsection, both men and women have popularity preferences. We say that popularities are \emph{symmetric} when for every acceptable pair $(m,w)$ we have $\D_w(m) = \D_m(w)$. Symmetric popularities model ``cross-sided'' correlations, for example when men and women prefer partners with whom they share some centers of interest. Interestingly, symmetric popularities can also encompass ``one-sided'' correlations: defining $\D_m(w) = \D_w(m) = \D(m) \cdot \D(w)$ is equivalent (because of the renormalization of popularities) with men and women having an intrinsic popularity specified by $\D$.
\medbreak
To measure the extent to which popularities are symmetric, we introduce a parameter $r \geq 1$, such that for each acceptable pair $(m,w)$ the values of $\D_w(m)$ and $\D_m(w)$ are within a factor $r$ of each other. 

\begin{theorem}
\label{thm:sympop}
Let $r \geq 1$ be a parameter. Assume that each person $p$ has popularity preferences defined by $\D_p$ over a set of $d_p \geq 1$ acceptable partners, such that if man $m$ and woman $w$ are mutually acceptable then the values of $\D_w(m)$ and $\D_m(w)$ are within a factor $r$ of each other. Then:
\[\E[\text{Number of stable pairs}] \leq N(1+\ln r) +\sum_{w\in \W}\frac{\ln d_w}{2} +\sum_{m\in \M} \frac{\ln d_m}{2} \]
When $r=1$, popularities are symmetric and the number of stable pairs is at most $N+N\ln N$.
\end{theorem}

\begin{proof}\setcounter{equation}{0}
For each person $p$, define $X_p$ the event where $p$ is matched.
We write $Y$ the total number of stable pairs.
Using the fact that women have popularity preferences, we start with the same proof as Theorem~\ref{thm:idpop}.
\begin{equation}\label{thm:sympop:1}
\E[Y] \leq \sum_{w\in\W}(1 + \ln d_w)\cdot \P[X_w] + \sum_{m\in\M} \E\bigg[\ln \frac{\D_{\mu_\W(m)}(m)}{\D_{\mu_\M(m)}(m)}\;|\;X_m\bigg]\cdot \P[X_m]
\end{equation}
Symmetrically, we can use the fact that men have popularity preferences, and use a symmetric version of Theorem~\ref{thm:popularity} to bound the expected number of stable wife of each man.
\begin{equation}\label{thm:sympop:2}
\E[Y] \leq \sum_{m\in\M}\left(1 + \ln d_m +\E\bigg[\ln \frac{\D_m(\mu_\M(m))}{\D_m(\mu_\W(m))}\;|\;X_m\bigg]\right)\cdot \P[X_m]
\end{equation}
Summing equations (\ref{thm:sympop:1}) and (\ref{thm:sympop:2}) yields
\begin{align*}
2\E[Y] &\leq \sum_{w\in\W}(1+\ln d_w)\cdot\P[X_w]+\sum_{m\in\W}(1+\ln d_m)\cdot\P[X_m]\\
&+\sum_{m\in\M} \E\left[\ln\frac{\D_{\mu_\W(m)}(m)}{\D_m(\mu_\W(m))}+\ln\frac{\D_m(\mu_\M(m))}{\D_{\mu_\M(m)}(m)} \;|\;X_m\right]\cdot \P[X_m]
\end{align*}
To conclude the proof, observe that the sums $\sum_w \P[X_w]$ and $\sum_m \P[X_m]$ are at most $N=\min(M,W)$, and that all ratios $\D_w(m)/\D_m(w)$ and  $\D_m(w)/\D_w(m)$ can be bounded by $r$.
\end{proof}

\section{Proof of Theorem \ref{theorem:womanbounded}}
\label{section:bounded_proofs}


Theorem~\ref{theorem:womanbounded} provides a bound on the maximal ratio of popularity of two stable husbands of a woman.
The bound depends on how uniform the preferences of men are (parameter $R_\M$)
and how similar the preferences of women are (parameter $Q_\W$).

\womanboundedtheorem*


\begin{proof}[Proof of Theorem~\ref{theorem:womanbounded}]
Fix a woman $w$ and denote $\mu_\M$ be the man-optimal stable matching.
In the sequel, we reference as "$w$-popularities" the popularities from the point of view of $w$ given by  $\D_w$.
We also fix the preferences of all women except $w$.

\smallskip

The proof relies on the computation of the set of stable husbands of $w$ by Algorithm~\ref{algo:MPDAextended}. 

The algorithm is presented as running deterministically, with all preferences of men and women given as input, chosen \emph{ex-ante} before the algorithm starts;
however in the case of popularity preferences (Definition~\ref{definition:popularity}) this algorithm can also be seen as a stochastic process.

In phase 1 of Algorithm~\ref{algo:MPDAextended},
instead of computing the lists of men \emph{ex-ante},
we might instead disclose them progressively along the execution of the algorithm. 
When a man is about to propose, he randomly picks a woman among those to whom he has not proposed yet, following a lottery whose probabilities are proportional to the popularities of the remaining women.
When a woman $\neq w$ receives a proposal, her answer is deterministic, consistent with her preferences fixed ex-ante. When $w$ receives a proposal, she refuses it.

In phase 2 of Algorithm~\ref{algo:MPDAextended},
woman $w$ parses the list $s$ of proposals she received.
She accept a proposal from a man of 
popularity $p$ with probability $\frac{p}{p + p_\bot}$, where $p_\bot$ denotes the sum of the popularities of men from whom she has received a proposal so far.

This transition probabilities define a stochastic process on the possible states of 
Algorithm~\ref{algo:MPDAextended}. 
By definition of popularity preferences, a particular execution of the algorithm has the same probability to occur in this stochastic process as 
in the case where the input data is generated \emph{ex-ante}. 

In this stochastic process,
we focus on the sequence  $H_0,H_1,...$ of men that are enumerated by Algorithm~\ref{algo:MPDAextended}, after the initial computation of the man-proposing stable matching $\mu_\M$. The sequence includes all men doing a proposal, including proposals to $w$ but also to other women. A man appears in the list as many times as he makes a proposal.
The set of stable husbands of $w$ is exactly the set of men from whom $w$ accepts such a proposal.

\smallskip

We define the notion of \emph{standard preferences} for women $\neq w$.
Set 
\[
T = |N|^5\enspace.
\]
Let $h,h'$ be two men and $w'$ a woman
such that $h'$ is $T$ times more popular than $h$ but, still, $w'$ prefers $h$ to $h'$.
     We say that the preferences of women $\neq w$ are \emph{standard} if no such tuple $h,h',w'$ exists.
Given $h,h'$, there is probability  $\frac{1}{1+T}$ that $w'$ prefers $h$ to $h'$, thus there is probability  $\leq \frac{1}{N^2}$ that preferences of women are not standard.

  Since $w$ has at most $N$ stable husbands, the contribution of non-standard events to the expected number of husbands of $w$ is at most $\frac{1}{N}$.

\smallskip

Without loss of generality, we assume that men $m_1,m_2,\ldots,m_N$
are indexed by increasing $w$-popularity:
\[
\D_w(m_1) \leq \D_w(m_2) \leq \ldots \leq \D_w(m_N)\enspace.
\]


Let $i$ be the index of the husband of $w$ in the man-proposing stable matching
\[
\mu_\M(w) = m_{i}\enspace.
\]

We partition the set of men in sets of exponentially increasing sizes,
starting with all men less or equally popular than $m$.
Let $F_0=[1,m_i]$ ,$F_1=(m_i,2*m_i]$,
$F_2=(2*m_i,4*m_i]$,..., $F_j=(2^j*i,N]$,
where $2^j*i < N \leq 2^{j+1}*i$.
Note that
$
j \leq \ln_2\left(N\right)\enspace.
$

Set 
\begin{align*}
&L=\frac{4\ln(N)}{\ln\left(1 + \frac{1}{R_\M} \right)}
&K=(T\cdot Q_\W)^L\enspace.
\end{align*}

For every set $F_\ell,\ell\in 1\ldots k$,
we say that \emph{there is a huge popularity gap} in  $F_\ell$ if the popularity ratio for $w$ in this interval is $\geq K$, i.e. if $\D_w(m_{2^\ell*m_i}) \leq K \D_w(m_{\min(N,2^{\ell+1}*m_i})$.

The case where there is no huge popularity gap in any of the $F_\ell,\ell\in 1\ldots k$   is an easy one:
then the most $w$-popular man $m_N$ has $w$-popularity at most $K^{j+1}
 \D_w(m_{i}) $
hence the conclusion of the theorem 
since $
j \leq \ln_2\left(N\right)
$.

Otherwise we select the smallest $\ell$ for which there is huge popularity gap in $F_\ell$.
Denote $E_0=F_0\cup F_1\cup\ldots F_{\ell-1}$
and $E_1=F_\ell$.
We show the following property:
\begin{itemize}
    \item[($\spadesuit$)]
    Fix the preferences of women $\neq w$ and assume these preferences are standard.
Then with probability $\geq 1 - \frac{1}{N^2}$, the $w$-popularities of men enumerated by Algorithm~\ref{algo:MPDAextended} are less or equal to $K$ times the $w$-popularity of
the most $w$-popular man in $E_0$.
\end{itemize}

\smallskip

Before turning to the proof of $(\spadesuit)$,
we show how to use it in order to complete the proof of the theorem.
Assume property ($\spadesuit$) holds.
Women $\neq w$ have standard preferences
with probability  $\geq 1 - \frac{1}{N^2}$.
By union bound, with probability  $\geq 1 - \frac{2}{N^2}$ the conclusion of $(\spadesuit)$ holds.
By minimality of $\ell$,
all men in $E_0$
have popularity
$\leq K^{ \ell-1}\D_w(m_{i})$.
Thus,
since $\ell \leq j\leq \ln_2(N)$
all men enumerated by Algorithm~\ref{algo:MPDAextended}
have popularity
$\leq K^{ \ln_2\left(N\right)}\D_w(m_{i})$.
Then with probability 
$\geq (1 - \frac{2}{N^2})$ the popularities of men proposing during the execution of Algorithm~\ref{algo:MPDAextended} are bounded
by 
\begin{equation}\label{eq:rrhugo}
K^{1 + \ln_2\left( N
\right)}
\enspace.
\end{equation}
By hypothesis about standard preferences,
$w$ for sure cannot have any stable husband $T$ times less popular than $m_i$,
hence the theorem.

\smallskip

To prove $(\spadesuit)$,
we assume the preferences of women $\neq w$ are standard
and bound the $w$-popularities of men in the sequence $H_0,H_1,...$.
We rely on the following properties of this sequence.
Denote $E_1'\subseteq E_1$ the set of men in $E_1$
whose $w$-popularity is above $T\cdot Q_\W$ times the least popular men in $E_1$ and below $T\cdot Q_\W$ times the most popular man in $E_1$.

\begin{enumerate}
\item[i)] the $w$-popularities of two consecutive men in the sequence $H_0,H_1,...$ may increase by a multiplicative factor of at most $T\cdot Q_\W$.
\item[ii)]
Assume $H_i\in E_0$.
Then $H_{i+1}\in E - E_1'$.
\item[iii)]
Assume that in the current execution of the algorithm, the proposing man is $H_i \in E_1'$. The probability that  $H_{i+1}$ belongs to $E_0$ conditioned on the current execution is 
$\geq \frac{1}{R_\M+1}$.
\item[iv)]
The probability that there exists $i$ such that the sequence 
$H_i,H_{i+1},...$ stays for more than $L$ consecutive steps in $E_1'$ is $\leq \frac{1}{N^2}$. \end{enumerate}
Properties i) - iii) follow from the hypothesis that women $\neq w$ have standard preferences: for a woman $w'\neq w$ married to $H_{i}$ to accept a proposal from $H_{i+1}$ it is necessary (but not sufficient in general) that 
$\frac{\D_{w'}(H_i)}{\D_{w'}(H_{i+1})} < T$.
In which case
$\frac{\D_{w}(H_i)}{\D_{w}(H_{i+1})} < T\cdot Q_\W$\enspace.
And if $w'=w$ then $w$ will anyway refuse the proposal from $H_{i+1}$ thus $H_{i+2}=H_{i+1}$.
Hence i) and ii). To prove iii), remark that whenever some man $H_i$ from $E'_1$ proposes to the wife of a man $h$, the wife will refuse (resp. accept) for sure if $h$ is outside $E_0\cup E_1$ (resp. is in $E_0$), because in that case $h$ is at least $T\cdot Q_\W$ times more $w$-popular (resp. less $w$-popular) than $H_i$. Since $E_0$ contains at least half of the men in $E_0\cup E_1$ and since women married to men in $E_0$ are at most $R_\M$ times less popular for $H_i$ than women married to men  in $E_1$, we get iii).
To prove iv), applying iii) repeatedly
shows that for every $i$ the stochastic process $H_i,H_{i+1},\ldots$ may stay in $E_1'$ for more than $L$ consecutive steps  with probability $\leq \left(1 - \frac{1}{R_\M+1}\right)^L
= \exp\left(-\ln(1 + \frac{1}{R_\M})\cdot L\right)= \frac{1}{|N|^4}$. Since the length of the sequence $H_0,H_1,\ldots$ is bounded by $N^2$, we get iv) by union bound.

To sum up, assuming women $\neq w$ have standard preferences, there is probability $\geq 1 - \frac{1}{|N|^2}$
to see at most $L-1$ consecutive steps in $E_1'$.
In this case,
according to properties i) and ii),
all men enumerated by Algorithm~\ref{algo:MPDAextended}
have popularities
at most 
$(T\cdot Q_\W)^{L-1}$
times the smallest $w$-popularity in $E_1'$, 
thus at most 
$(T\cdot Q_\W)^{L}$
times the largest $w$-popularity in $E_0$.
Since $K=(T\cdot Q_\W)^{L}$,
this proves ($\spadesuit$).
\end{proof}


The proof of Corollary~\ref{cor:womanbounded} relies on the following lemma.

\womanboundedlemma*

\begin{proof}
In the last phase of Algorithm~\ref{algo:MPDAextended},
$w$ enumerates the list of proposal
she has received.
We can remove from the input data of Algorithm~\ref{algo:MPDAextended}
the preferences of $w$ among men she does not have received a proposal yet.
Then the last phase of Algorithm~\ref{algo:MPDAextended}
can be turned to a stochastic process
where $w$ progressively reveals her preferences.
 Lemma~\ref{lemma:popularity} specifies how the transition probabilities of this process
 should be defined in order to guarantee than
 an execution of the algorithm has the same probability in the two versions (ex-ante vs online revelation of the preferences of $w$ on $L$).
 
According to Lemma~\ref{lemma:popularity},
the probability that $w$ accepts a proposal
from a man $m$ of the list depends on:
the popularity $p$ of $m$; and
the sum $p_\bot$  of the popularities of men from whom she has received a proposal during the initial execution of Algorithm~\ref{algo:MPDA} or which appear 
before $m$ in the list $L_w$.
Then $w$ accepts the proposal of $m$ with probability $\frac{p}{p + p_\bot}$.


Let $p_0$ be the popularity for $w$ of the initial husband $\mu_\M(w)$ of $w$.
Let $p_\bot$ be the sum of the popularities received by $w$ during the computation of the man-proposing stable matching.
Let $p_1,\ldots,p_{|L_w|}$
be the popularities (for $w$) of the men
in the list $L_w$.
Set 
\[
R = \max_{i \in 1\ldots|L_w|} \frac{p_i}{p_0}\enspace.
\]
Then, using the sum-integral comparison~\eqref{eq:1} in the proof of Theorem~\ref{thm:popularity},
\begin{align*}
\sum_{i=1}^{|L_w|} \P[\hbox{proposal $x_i$ is accepted by $w$}]
&\leq
\ln(p_\bot + p_1 +\dots+p_{|L_w|}) - \ln p_\bot\\
&\leq
\ln(p_\bot + p_1 +\dots+p_{|L_w|}) - \ln p_0\\
&
\leq 
\ln(1 + p_1/p_0 +\dots+p_{|L_w|}/p_0)\\
&\leq
\ln(R|L_w|)
=
\ln(|L_w|)
+\ln(R)
\enspace.
\end{align*}
\end{proof}

\womanboundedcor*

\begin{proof}
By inspection of the proof,
remark that the bound given by Theorem~\ref{theorem:womanbounded}
is actually a bound on the popularities of proposals received by $w$ during the execution of Algorithm~\ref{algo:MPDAextended} (cf~\eqref{eq:rrhugo}).
We conclude with Lemma~\ref{lem:prel}.
\end{proof}

{\color{orange}






}

\end{document}